\documentclass[11pt,letterpaper]{article}
\usepackage[margin=1.1in]{geometry}
\usepackage{amsmath,amssymb,amsfonts,setspace,url}
\usepackage{amsthm,hyperref}
\usepackage{color}
\usepackage{cleveref} %easy referencing via \Cref
\usepackage{ascmac}
\usepackage{mathpazo}
\makeatletter
\newcommand{\newreptheorem}[2]{\newtheorem*{rep@#1}{\rep@title}\newenvironment{rep#1}[1]{\def\rep@title{#2 \ref*{##1}}\begin{rep@#1}}{\end{rep@#1}}}
\makeatother
\makeatletter
\newcommand{\newconjecture}[2]{\newconjecture*{rep@#1}{\rep@title}\newenvironment{rep#1}[1]{\def\rep@title{#2 \ref*{##1}}\begin{rep@#1}}{\end{rep@#1}}}
\makeatother
\theoremstyle{plain}
\newtheorem{theorem}{Theorem}[section]
\newtheorem{lemma}[theorem]{Lemma}
\newtheorem{conjecture}{Conjecture}
\newtheorem{definition}[theorem]{Definition}
\newtheorem{fact}{Fact}
\newtheorem{remark}[theorem]{Remark}
\newtheorem{proposition}[theorem]{Proposition}
\theoremstyle{definition}

\newcommand{\norm}[1]{\|#1\|}

\newcommand{\Comp}{\mathbb{C}}
\newcommand{\myO}{O^\ast\!}
\newcommand{\Real}{\mathbb{R}}
\newcommand{\Nat}{\mathbb{N}}
\newcommand{\ket}[1]{|#1\rangle}

\newcommand{\bra}[1]{\langle#1|}
\newcommand{\BQP}{\mathsf{BQP}}
\newcommand{\NP}{\mathsf{NP}}
\newcommand{\QMA}{\mathsf{QMA}}
\newcommand{\QCMA}{\mathsf{QCMA}}

\newcommand{\MA}{\mathsf{MA}}

\newcommand{\poly}{\mathrm{poly}}

\newcommand{\real}{\mathsf{real}}
\newcommand{\rep}{\mathsf{rep}}
\newcommand{\comp}{\mathsf{imag}}
\providecommand{\Aa}{\mathcal{A}}
\providecommand{\Bb}{\mathcal{B}}
\providecommand{\Cc}{\mathcal{C}}
\providecommand{\Ff}{\mathcal{F}}
\providecommand{\Qq}{\mathcal{Q}}

\providecommand{\Sq}{\mathcal{S}\!\mathcal{Q}}

\providecommand{\Tt}{\mathcal{T}}

\newcommand{\E}{\mathbb{E}}
\newcommand{\var}{\mathrm{var}}

\newcommand{\LH}{\mathsf{LH}} %local hamiltonian
\newcommand{\GLHest}{\mathsf{GLH}^{\mathsf{est}}}
\newcommand{\GLHdes}{\mathsf{GLH}^{\ast}} % guided local hamiltonian
\newcommand{\GLH}{\mathsf{GLH}}% promise local hamiltonian
\newcommand{\EST}{\mathsf{SVT}}
\newcommand{\SV}{\mathsf{\mathsf{SVE}}}

\newcommand{\Piyes}{\Pi_{\textup{yes}}}
\newcommand{\Pino}{\Pi_{\textup{no}}}
\newcommand{\abs}[1]{\left|#1\right|}
\newcommand{\set}[1]{{\left\{#1\right\}}}    % braces for set notation
\newcommand{\hin}{H_{\textup{in}}}
\newcommand{\hprop}{H_{\textup{prop}}}
\newcommand{\hstab}{H_{\textup{stab}}}
\newcommand{\hout}{H_{\textup{out}}}

\newcommand{\trnorm}[1]{\norm{#1}_{\mathrm {tr}}}  % trace norm
\newcommand{\reals}{{\mathbb R}}

\def\ket#1{ | #1 \rangle}
\def\bra#1{{\langle #1 | }}
\newcommand{\ketbra}[2]{\ket{#1}\!\bra{#2}}        % outer product
\newcommand{\braket}[2]{\mbox{$\langle #1  | #2 \rangle$}}

\newcommand\lmin[1]{\lambda_{#1}}
\newcommand{\psihist}{\psi_{\textup{hist}}}
\newcommand{\pQMAlog}{\textup{P}^{\textup{QMA}[\log]}}
\newcommand{\spa}[1]{\mathcal{#1}}

\begin{document}

\title{Dequantizing the Quantum Singular Value Transformation: \\
Hardness and Applications to Quantum Chemistry and the Quantum PCP Conjecture}
%\title{Theoretical Evidence of the Superiority of Quantum Algorithms for Chemistry}
\author{
Sevag Gharibian\\
Paderborn University\\
sevag.gharibian@upb.de
\and
Fran{\c c}ois Le Gall\\
Nagoya University\\
legall@math.nagoya-u.ac.jp
}

\date{}
\maketitle
\thispagestyle{empty}
\setcounter{page}{1}
\begin{abstract}
The Quantum Singular Value Transformation (QSVT) is a recent technique that gives a unified framework to describe most quantum algorithms discovered so far, and may lead to the development of novel quantum algorithms. In this paper we investigate the hardness of \emph{classically} simulating the QSVT. A recent result by Chia, Gily{\'{e}}n, Li, Lin, Tang and Wang (STOC 2020) showed that the QSVT can be efficiently ``dequantized'' for low-rank matrices, and discussed its implication to quantum machine learning. In this work, motivated by establishing the superiority of quantum algorithms for quantum chemistry and making progress on the quantum PCP conjecture, we focus on the other main class of matrices considered in applications of the QSVT, sparse matrices.

We first show how to efficiently ``dequantize'', with arbitrarily small \emph{constant precision}, the QSVT associated with a low-degree polynomial. We apply this technique to design classical algorithms that estimate, with constant precision, the singular values of a sparse matrix. We show in particular that a central computational problem considered by quantum algorithms for quantum chemistry (estimating the ground state energy of a local Hamiltonian when given, as an additional input, a state sufficiently close to the ground state) can be solved efficiently with constant precision on a classical computer. As a complementary result, we prove that with \emph{inverse-polynomial precision}, the same problem becomes $\BQP$-complete. This gives theoretical evidence for the superiority of quantum algorithms for chemistry, and strongly suggests that said superiority stems from the improved precision achievable in the quantum setting.  We also discuss how this dequantization technique may help make progress on the central quantum PCP conjecture.
%, by establishing a novel version of this conjecture that may be easier to prove.
\end{abstract}
\newpage
%=====================================
\section{Introduction}
%=====================================
%=================
\subsection{Motivations}\label{sub:motivations}
%=================
\paragraph*{The Quantum Singular Value Transformation.}
The Quantum Singular Value Transformation (QSVT) is a recent technique developed by Gily{\'{e}}n, Su, Low and Wiebe \cite{Gilyen+STOC19}, which both generalizes the Hamiltonian simulation results by Low and Chuang \cite{Low+17} and exploits the framework of ``qubitization'' introduced in \cite{Low+19} (see also \cite{Chakraborty+ICALP19}). Informally, the QSVT associated to a function $f\colon[0,\infty)\to\Real$ can be interpreted as a process that implements, given access to a matrix $A\in \Comp^{M\times N}$, a unitary operator ``representing'' the action of the (generally non-unitary) matrix $f(A)$. The matrix $f(A)$ is defined via the singular value decomposition of~$A$, which explains why this transformation is called the quantum singular value transformation. The QSVT gives a unified framework to describe most quantum algorithms discovered so far, such as Hamiltonian simulation \cite{Lloyd96}, Grover search \cite{GroverSTOC96}, amplitude amplification \cite{Brassard+02}, the HHL algorithm for solving linear systems of equations \cite{HHL09} and some quantum walks \cite{Magniez+SICOMP11,SzegedyFOCS04}. As shown in \cite{Gilyen+STOC19}, the QSVT can lead to the development of novel quantum algorithms as well. We also refer to \cite{Martyn+21} for a survey of all these techniques.

Chia, Gily{\'{e}}n, Li, Lin, Tang and Wang \cite{Chia+STOC20} and Jethwani, Le Gall and Singh \cite{Jethwani+MFCS20} have recently shown that under mild smoothness conditions on the function $f$ (which are satisfied, for instance, if $f$ is a low-degree polynomial function), quantum algorithms computing the QSVT of a matrix~$A$ can be ``dequantized'' if $A$ is a low-rank matrix\footnote{More accurately, the runtimes of \cite{Chia+STOC20,Jethwani+MFCS20} depend polynomially on the Frobenius norm of $A$, which is small when the rank is small.} and $A$ is accessible via $\ell_2$-norm sampling. The latter assumption, dating back to works on randomized algorithms for linear algebra \cite{Frieze+JACM04}, means that rows and columns of $A$ can be sampled according to their $\ell_2$-norm, so that rows and columns that are heavier in~$\ell_2$-norm will be sampled more frequently. Starting from Tang's breakthrough on dequantizing quantum algorithms for recommendation systems \cite{TangSTOC19}, this kind of access to the data has now become a standard assumption when considering dequantization of quantum algorithms for linear algebraic problems. Assuming that the input matrix $A$ has low rank, on the other hand, is a strong assumption. This assumption is motivated by applications to machine learning. Indeed, based on this dequantization technique, Chia, Gily{\'{e}}n, Li, Lin, Tang and Wang \cite{Chia+STOC20} were able to dequantize most of the quantum algorithms for machine learning developed so far.

Another central class of matrices considered in quantum algorithms, and especially in quantum algorithms for the simulation of quantum systems, is sparse matrices. Such matrices play a fundamental role in quantum complexity theory as well: the quantum PCP conjecture \cite{Aharonov+13}, which is one of the central conjectures in quantum complexity theory, and specifically in the field of Hamiltonian complexity (see, e.g., \cite{O11,Gharibian+15}), mainly focuses on a subclass of sparse matrices called \emph{local Hamiltonians} (Hamiltonians where each local term only acts on a small number of $k$ of quantum particles, say $k\in O(1)$). In general, it is \emph{not} expected that the QSVT can be dequantized for sparse matrices and \emph{all} (smooth) functions $f$--- this is because this setting contains the HHL~\cite{HHL09} algorithm for inverting well-conditioned matrices, a problem known to be $\BQP$-complete for (high-rank) sparse matrices,
%\footnote{In the special case of \emph{low-rank} matrices, the HHL algorithm {can} be dequantized~\cite{Chia+20}.},
even for \emph{constant} precision~\cite{HHL09}. It is thus fundamental to understand more thoroughly for which functions $f$ the QSVT can be dequantized, and for which $f$ we have a provable quantum advantage.

\paragraph*{Quantum chemistry and eigenvalue estimation.} Quantum chemistry is considered one of the most promising applications of quantum computers \cite{Aaronson09, Aspuru-Guzik+05, Bauer+20, Lee+21, Reiher+17}. From a theoretical perspective, the main goal in quantum chemistry algorithms is clearly defined: compute an estimate of the ground state energy of a local Hamiltonian representing the quantum system\footnote{Most, if not all, known quantum many-body systems in nature evolve in time according to local Hamiltonians, which are $2^n\times 2^n$ Hermitian matrices specified succinctly with $\poly(n)$ bits via local constraints, analogous to local clauses for $3$-SAT. The spectrum of said Hamiltonians is particularly important, as each distinct eigenvalue encodes an energy level of the system, and each corresponding eigenspace the state of the system at that energy. The smallest eigenvalue and corresponding eigenvector are called the \emph{ground state energy}  and \emph{ground state}, respectively.} or, in more mathematical terms, compute the smallest eigenvalue of a sparse matrix. The most rigorous quantum approach to speed up such calculations, first proposed in \cite{Abrams+99,Aspuru-Guzik+05}, utilizes the quantum phase estimation technique \cite{Kitaev95}. Other approaches such as variational quantum algorithms, expected to be easier than phase estimation to implement on near-term quantum computers, have also been developed (see \cite{Cerezo+21} for a survey of recent works on variational quantum algorithms). These other approaches, however, are mostly heuristic-based and thus their performance is much more difficult than quantum phase estimation to analyze rigorously.

\paragraph*{Quantum Phase Estimation and quantum chemistry.} The standard problem quantum phase estimation (QPE) solves is roughly as follows: given (1) an efficient quantum implementation of a unitary matrix $U\in\Comp^{N\times N}$ and (2) an eigenvector $\ket{u}\in\Comp^N$ of $U$ encoded as a physical quantum state, output an estimate of the eigenvalue\footnote{The name ``phase'' estimation arises since all eigenvalues of unitary $U$ are of form $e^{i\theta}$ for $\theta\in\reals$, where $\theta$ is called the ``phase''.} $\lambda_u$ satisfying $U\ket{u}=\lambda_u\ket{u}$. Since any unitary can be written $U=e^{iH}$ for some Hermitian matrix $H$, one can also apply QPE to estimate eigenvalues of \emph{Hermitian} $H$ (assuming one can efficiently implement powers of the unitary $e^{iH}$, which is possible for sparse $H$~\cite{AT03,BCCKS14,Low+19}).

Thus, in quantum chemistry, $U=e^{iH}$ encodes the Hamiltonian $H$ of the system, which is a sparse matrix, and we are interested in the smallest eigenvalue (the ground state energy) of~$H$. In general, this problem (denoted the \emph{local Hamiltonian problem}) is well-known to be $\QMA$-complete \cite{Kitaev+02}, for Quantum Merlin-Arthur ($\QMA$) a quantum analogue of $\NP$. Thus, an additional central working assumption\footnote{We stress that this is a strong assumption. First, from a complexity-theory perspective, Hamiltonians occurring in chemistry can be $\QMA$-complete \cite{Schuch+09}. Thus, a uniformly generated efficient quantum circuit approximating the ground state would yield $\BQP=\QMA$, which is highly unlikely. Indeed, even a \emph{non-uniformly} generated efficient quantum circuit would yield $\QCMA=\QMA$, for $\QCMA$ defined as $\QMA$ but with a classical proof. From a practical perspective, it is also not expected that one can \emph{always} produce an initial state with at least inverse-polynomial overlap with the ground state. In fact, it is possible to contrive instances where the initial state produced by the Hartree-Fock method (described later) has support on the ground state that vanishes exponentially in the number of particles. A classic example is the ``Van Vleck catastrophe'' (see \cite{McClean+14}).} is needed in quantum chemistry: One assumes the ability to efficiently prepare a quantum state ``sufficiently close'' to the true ground state. This assumption is motivated by the observation that in practice, computationally-efficient classical techniques in computational quantum chemistry often give a fairly good initial approximation of the ground state. For instance, the Hartree-Fock method \cite{Echnique+07} typically\footnote{This statement requires clarification. Roughly, the Hartree-Fock method may be viewed as ``optimization over a product state ansatz'' in the second quantization formalism. More formally, one wishes to compute the Fock state with minimal energy against a given fermionic Hamiltonian. For such approaches, however, it should be clarified that even if a Fock state achieves a good approximation of the ground state energy, said state does \emph{not} necessarily achieve good overlap with the ground state. This phenomenon appears for qubit-systems as well --- dense local Hamiltonian systems can be efficiently solved to \emph{arbitrary} constant relative error $\epsilon$ via a product state ansatz~\cite{GK12_3,BH13,B22}, even though the actual ground state may be highly-entangled, and thus far from any product state.} recovers 99\% of the total energy \cite{Whitfield+13}.

With this assumption in place, the question is now about \emph{precision} in QPE. As discussed in \cite{Lee+21}, the holy grail in quantum chemistry is to estimate the ground state energy with precision less than the ``chemical accuracy'' (about 1.6 millihartree), where the norm of the Hamiltonian is allowed to increase polynomially in the number of particles and local dimension per particle. However, if the Hamiltonian is normalized (as in the current paper), then the required precision is inverse-polynomial in these parameters.
The key point is now that QPE has two advantages: (1) It can produce an eigenvalue estimate within inverse polynomial precision efficiently, and (2) this works (with polynomial overhead) even when QPE is given only a ``rough'' approximation~$\ket{\tilde u}$ of the ground state~$\ket{u}$. Thus, quantum algorithms can solve the desired chemistry problem to the desired accuracy in poly-time, given a ``sufficiently good'' initial state $\ket{\tilde u}$.

\paragraph*{Connecting QPE with the QSVT.} QPE can be cast as a special case of the QSVT (see, for instance, Sections IV and V of \cite{Martyn+21}). Thus, given the quantum chemistry application above, a natural question is whether the QSVT, when applied to a local Hamiltonian $H$ or, more generally, to a sparse matrix $A\in\Comp^{N\times N}$, can be dequantized. Somewhat unexpectedly, this is (trivially) possible under mild conditions if access to the \emph{exact} eigenvector is provided: if we have classical random-access to the ground state $u\in\Comp^{N}$ and can efficiently  find a non-zero entry of $u$ (which can be done, for instance, if $u$ is also accessible via $\ell_2$-norm sampling), then the ground state energy can be computed exactly, classically, by computing only one entry of the vector $Au$ (for instance, if $u_1\neq 0$ then $(Au)_1/u_1$ is the eigenvalue corresponding to $u$), which can be done efficiently when $A$ is sparse, e.g., when~$A$ has only $O(\poly(\log N))$ non-zero entries per row. Of course, having an {exact} description of~$u$ in the context of quantum chemistry is not a reasonable assumption, so the interesting and practically relevant question is: \emph{Do similar efficient classical strategies exist given just an approximation $\tilde u$ to $u$?} To our knowledge, this question, closely related to dequantizing the QSVT and central to establishing theoretical foundations of quantum algorithms for quantum chemistry, has not been studied in the literature.\footnote{Wocjan and Zhang \cite{Wocjan+06} showed the $\BQP$-completeness of sampling with precision $1/\poly(\log(N))$ from the eigenvalues of the matrix $A$. This problem is different to our problem: we are targeting one specific eigenvalue (the smallest one) and we are given as input a vector that has large overlap with the corresponding eigenvector (this is a non-trivial assumption which may significantly aid the estimation of the eigenvalue).}

%==================
\subsection{Main results}\label{sub:mainres}
%==================
We first describe our results concerning estimating the ground state energies of local Hamiltonians, which are easier to state and most closely related to quantum chemistry and quantum PCP applications (see Section \ref{sub:app} for applications). We then describe our more general technique to dequantize the QSVT and estimate singular values of sparse matrices. To begin, we define the notions of ``query-access'' and ``sampling-access''.

\paragraph*{Query-access and sampling-access.}
In this work, as in prior dequantization works\footnote{Rudi, Wossnig, Ciliberto, Rocchetto, Pontil, and Severini~\cite{Rudi2020approximating} also consider ``dequantization'' of simulating quantum Hamiltonian via sampling assumptions. A key difference to our work is that their runtime scales polynomially in the sparsity and Frobenius norm of the Hamiltonian, whereas our runtime is independent of the Frobenius norm (and thus of the rank).} \cite{vdN11,SvdN13,TangSTOC19,Chia+STOC20,Chia+20,Gilyen+20,Jethwani+MFCS20,Rudi2020approximating}, we consider two ways to access vectors and matrices. For a vector $u\in\Comp^N$ for some integer $N\ge 1$,
we say we have \emph{query-access} to $u$ if for any index $i\in\{1,\ldots,N\}$, we can compute coordinate $u_i\in\Comp$ in $O(\poly(\log N))$ time (we assume in this paper that scalars, i.e., complex numbers, can be manipulated at unit cost). Query-access thus corresponds to standard polylogarithmic-cost RAM access to the vector $u$. Query-access to a matrix $A\in\Comp^{N\times N}$ is defined similarly (if the matrix is sparse, we assume that we can directly access the non-zero entries --- see Definition~\ref{def:sparsequery} in Section \ref{sec:prelim}). We say that we have \emph{sampling-access} to $u$ if in addition to having query-access to $u$, we can also efficiently sample from the probability distribution that outputs index $j$ with probability $|u_j|^2/\norm{u}^2$, where $\norm{\cdot}$ is the Euclidean norm, for each $j\in\{1,\ldots,N\}$. For technical reasons, we also require knowledge of a good estimate of $\norm{u}$. Our actual definition of sampling-access is even more general, since it only requires the ability to sample from a distribution close to the ideal distribution, and also only requires the knowledge of an approximation of $\norm{u}$. We refer to Definition~\ref{def:sample} for the precise definition.

\paragraph*{Estimating the ground state energy of a local Hamiltonian.} Given a Hamiltonian $H$ acting on~$n$ qubits (i.e., a Hermitian matrix $H\in\Comp^{2^n\times 2^n}$), let $\lambda_H$ denote its ground state energy (i.e., the smallest eigenvalue of $H$) and $V_H$ the vector space spanned by its ground states (i.e., the vector space spanned by all the eigenvectors corresponding to $\lambda_H$). We denote by $\Pi_H$ the orthogonal projection onto $V_H$. One says $H$ is $k$-local if it can be written as a sum of $\poly(n)$ terms $H=\sum_{i=1}^{\poly(n)}H_i$, where each term $H_i$ acts non-trivially on at most $k$ qubits.
In this paper, we consider the problem of computing the ground state energy $\lambda_H$ of a local Hamiltonian when ``guided'' by a vector $u$ that has some overlap with the space $V_H$. Since in quantum algorithms for this problem the vector $u$ is given as a quantum state, in this paper we consider the natural classical analog: we assume we have sampling-access to $u$. We denote this problem by $\GLHest(k,\epsilon, \delta)$, where $k\ge 2$ is an integer and $\epsilon,\delta$ are real numbers in the interval $(0,1]$. The description of this problem is as follows.\footnote{This problem has similarities with the ``Guided Stoquastic Local Hamiltonian'' problem introduced by Bravyi \cite{Bravyi2015}. In the Guided Stoquastic Local Hamiltonian problem the main differences are: (1) only stoquastic Hamiltonians are considered; (2) in the YES case, an ``easy guiding state'' exists, but in the NO case, no assumption on guiding states is made; (3) the guiding state is not given as an input. The definition of ``guiding state'' in \cite{Bravyi2015} also differs from ours, requiring non-negative amplitudes and ``point-wise overlap'' with the entries of the ground state, as opposed to a global overlap here.}

\begin{center}
\fbox{
\begin{minipage}{13 cm} \vspace{2mm}

\noindent$\GLHest(k,\epsilon,\delta)$ \hspace{3mm}(guided local Hamiltonian)\\\vspace{-3mm}

\noindent\hspace{3mm} Input: a $k$-local Hamiltonian $H$ acting on $n$ qubits such that $\norm{H}\le 1$

\noindent\hspace{15mm}
%the classical description of a vector $u\in \Sample_n$ such that $\norm{u}\le 1$ and $\norm{\Pi_Hu}\ge \delta$
sampling-access to a vector $u\in\Comp^{2^n}$ such that $\norm{u}\le 1$
\vspace{2mm}

\noindent\hspace{3mm}
Promise: $\norm{\Pi_Hu}\ge \delta$
\vspace{2mm}

\noindent\hspace{3mm} Output: an estimate $\hat{\lambda}$ such that $|\hat\lambda-\lambda_H|\le \epsilon$
\vspace{2mm}
\end{minipage}
}
\end{center}

Quantum algorithms based on QPE (or the QSVT) can solve $\GLHest$ in $\poly(n)$ time for $k=O(\log n)$ even for $\epsilon=1/\poly(n)$ and $\delta=1/\poly(n)$, under the condition that a quantum state representing $u$ can be constructed efficiently (see for instance \cite[Section~IV]{Martyn+21} for a description of the decision version of the problem, which can then be converted into an algorithm for the estimation version by binary search).
In this paper, we first show that with constant precision and constant overlap, this problem can actually be solved in polynomial time even classically.

\begin{theorem}\label{th:classical}
For any constants $\epsilon,\delta\in(0,1]$ and any $k=O(\log n)$, the problem $\GLHest(k,\epsilon,\delta)$ can be solved classically with probability at least $1-1/\exp(n)$ in $O(\poly(n))$ time.
\end{theorem}
\noindent As described in the remark at the end of Section \ref{sub:proofthclassical}, the approach we develop to show Theorem~\ref{th:classical} actually also solves the version of the problem $\GLHest$ in which only approximate sampling-access to $u$ (as described in Definition \ref{def:sample}) is provided.

We complement this result by a hardness result. We actually show hardness even for a fairly simple class of guiding vectors, which we call \emph{semi-classical states\footnote{For clarity, there is no connection between our terminology ``semi-classical'' and (e.g.) semiclassical physics.}} and now define. First, recall the concept of subset state introduced by Grilo, Kerenidis and Sikora \cite{Grilo+16}:
for any non-empty subset $S \subseteq \{0,1\}^n$, the subset state associated with $S$ is the unit-norm vector
\[
\frac{1}{\sqrt{|S|}}\sum_{w\in S}\ket{w},
\]
i.e., the uniform superposition over the basis states corresponding to the strings in $S$. We say that a unit-norm vector $u\in\Comp^{2^n}$ is a semi-classical state if $u$ is a subset state associated with a subset $S\subseteq \{0,1\}^n$ with $|S|=\poly(n)$. Note that in this case the set $S$ can be represented using a polynomial number of bits. Thus, the \emph{description} of the semi-classical state $u$ is defined as this polynomial-length description of $S$, i.e., an explicit enumeration of the (polynomially many) elements of $S$.

%\begin{definition}\label{def:semiclassical}
%A unit-norm complex vector $u\in\Comp^{2^n}$, for some integer $n\ge 1$, that can be written in the form
%\[
%u=\ket{\overline{w_1}}\otimes\cdots\otimes\ket{\overline{w_n}}
%\]
%for a string $w\in\{0,1,2\}^n$, where $\ket{\overline{0}}=\ket{0}$, $\ket{\overline{1}}=\ket{1}$ and $\ket{\overline{2}}=\frac{1}{\sqrt{2}}(\ket{0}+\ket{1})$. The string $w$ is called the description of the vector $u$.
%\end{definition}
%\begin{definition}\label{def:semiclassical}
%Let $t$ be a positive integer. A $t$-term semi-classical state is a unit-norm complex vector $u\in\Comp^{2^n}$, for some integer $n\ge \log_2(t)$, that can be written in the form
%\[
%u=\frac{1}{\sqrt{m}}\sum_{\ell=1}^m\ket{w_\ell}
%\]
%for $t$ distinct strings $w_1,\ldots,w_t\in\{0,1\}^n$. The set of strings $\{w_1,\ldots,w_t\}$ is called the description of the vector $u$.
%\end{definition}
%\noindent
%Note that the definition is similar to the concept of subset state by Grilo, Kerenidis and Sikora \cite{Grilo+16}, except that we
Given the description of a semi-classical state $u$, sampling-access to $u$ can trivially be implemented. This observation motivates the definition of the following non-oracle version of (a gapped decision version of) the problem $\GLHest$, which we denote by $\GLHdes(k,a,b,\delta)$, where $k\ge2$ is an integer, $a,b\in[0,1]$ are two real numbers such that $a<b$, and $\delta\in(0,1]$.

\begin{center}
\fbox{
\begin{minipage}{14 cm} \vspace{2mm}

\noindent$\GLHdes(k,a,b,\delta)$ \hspace{3mm}(guided local Hamiltonian with semi-classical guiding vector) \\\vspace{-3mm}

\noindent\hspace{3mm} Input: $k$-local Hamiltonian $H$ acting on $n$ qubits such that $\norm{H}\le 1$

\noindent\hspace{15mm}
%the classical description of a vector $u\in \Sample_n$ such that $\norm{u}\le 1$ and $\norm{\Pi_Hu}\ge \delta$
description of a semi-classical state $u\in\Comp^{2^n}$
\vspace{2mm}

\noindent\hspace{3mm}
Promises: (i) $\norm{\Pi_Hu}\ge \delta$

\noindent\hspace{21mm} (ii) either $\lambda_{H}\le a$ or $\lambda_{H}\ge b$ holds
\vspace{2mm}

\noindent\hspace{3mm} Goal: decide which $\lambda_{H}\le a$ or $\lambda_{H}\ge b$ holds
\vspace{2mm}
\end{minipage}
}
\end{center}
%\snote{for both defs, should first mention that by ``classical description'' of $u$, we mean \Cref{def:query-access} or \Cref{def:sample}.}
%\fnote{I agree (I will polish the writing and the explanations later once we know exactly which definition we need to prove the hardness result). ``Classical description'' should refer to \Cref{def:sample}. Note that for the proof of $\BQP$-hardness this does not matter much, since the state that has large overlap with the ground state should be very easy to describe classically.}

%Given the description of a semi-classical state $u$, a quantum state representing $u$ can be efficiently created. The problem $\GLHdes(k,a,b,\delta)$ is thus in $\BQP$ for $k=O(\log n)$, even for $\epsilon=1/\poly(n)$ and $b-a=1/\poly(n)$.
Our hardness result is as follows.
\begin{theorem}\label{th:hardness}
For any $\delta\in (0,1/\sqrt{2}-\Omega(1/\poly(n)))$, there exist parameters $a,b\in[0,1]$ with $b-a=\Omega(1/\poly(n))$ such that $\GLHdes(6,a,b,\delta)$ is $\BQP$-hard.
\end{theorem}
\noindent Theorem \ref{th:hardness} shows that the decision version of the guided local Hamiltonian problem (and thus the estimation version as well) is $\BQP$-hard for inverse-polynomial precision, even for constant $k$ and $\delta$, and even when the guiding vector $u$ is a semi-classical state given explicitly. Combined with Theorem \ref{th:classical}, this gives compelling evidence that for the guided local Hamiltonian problem, the advantage of quantum computation mainly comes from the improved precision it provides.
%\fnote{Can we get $\delta$ arbitrarily close to $1$?}

\paragraph*{Dequantized algorithms for the Singular Value Transformation.} We now describe our more general technique to dequantize the QSVT and estimate singular values of sparse matrices.
For simplicity, here we discuss only the case of square matrices --- the general case of rectangular matrices is dealt with in Section \ref{sec:SVT}.

Any matrix $A\in \Comp^{N\times N}$ can be written in terms of its \emph{Singular Value Decomposition (SVD)},
\[
    A=\sum_{i=1}^{N}\sigma_i u_iv_i^\dagger,
\]
for real singular values $\sigma_i\geq 0$, and singular vectors $\{u_i\}_i$ and $\{v_i\}_i$, where each set forms an orthonormal basis for~$\Comp^N$. For any $a,b\in [0,\infty)$ such that $a<b$, let $\Pi_A^{[a,b]}$ denote the orthogonal projector onto the subspace generated by the vectors $v_i$ such that $\sigma_i\in[a,b]$.

The \emph{Singular Value Transformation (SVT)} of $A$ associated to an even polynomial $P\in \Real[x]$ (i.e., a polynomial such that $P(x)=P(-x)$) is defined as\footnote{Formally, $\abs{A}:=\sqrt{A^\dagger A}$ generalizes the absolute value function to matrix $A$, and ensures $\abs{A}$ is positive-semidefinite (which implies $\abs{A}$ is diagonalizable with real, non-negative eigenvalues). The polynomial $P$ is then applied as an ``operator function'' to $\abs{A}$, i.e., applied to the eigenvalues of $\abs{A}$.}
\begin{equation}\label{eq:SVTc}
P\left(\sqrt{A^\dagger A}\right)=\sum_{i=1}^{N}P(\sigma_i) v_iv_i^\dagger.
\end{equation}
We next say that a matrix is $s$-sparse if it contains at most $s$ non-zero entries per row and column.
For an $s$-sparse matrix $A\in\Comp^{N\times N}$ with $\norm{A}\le 1$, a unit-norm vector $u\in \Comp^N$ given as a quantum state, and a polynomial $P$ of degree $d$ such that $P(x)\in[0,1]$ for all $x\in[0,1]$, the framework for QSVT developed by Gily{\'{e}}n, Su, Low and Wiebe \cite{Gilyen+STOC19} gives a quantum algorithm that constructs a quantum representation of the vector $P(\sqrt{A^\dagger A})u$ in $O(ds\cdot \poly(\log N,\log s))$ time (assuming quantum query-access to $A$).
% \snote{this is presumably also assuming we have query-access to $A$?}.
%, which enables the implementation of sampling-access to the vector $p(\sqrt{A^\dagger A})u$.
For an arbitrary vector $v\in \Comp^{N}$ with $\norm{v}\le 1$ given via sampling-access, this quantum algorithm can then be used to estimate with high probability the complex number $v^\dagger P(\sqrt{A^\dagger A})u$ with additive precision $\epsilon$ in $\poly(\log N,d,s,1/\epsilon)$ time.

In this paper, we show the following result.
\begin{theorem}[Informal version of \Cref{th:classical-est-formal}]\label{th:classical-est}
Let $P\in\Real[x]$ be a constant-degree polynomial.
For any constant $\epsilon>0$, there is a $O(\poly(\log N,s))$-time classical algorithm that, given query-access to an $s$-sparse matrix $A\in\Comp^{N\times N}$ with $\norm{A}\le 1$, query-access to a vector $u\in\Comp^N$, and sampling-access to a vector $v\in \Comp^N$, outputs with high probability an estimate of the quantity $v^\dagger P(\sqrt{A^\dagger A})u$ with additive precision~$\epsilon$.
\end{theorem}
\noindent In words, Theorem \ref{th:classical-est} efficiently dequantizes the QVST of \cite{Gilyen+STOC19} for \emph{constant}-degree polynomials and \emph{constant} precision, in the sense that both the quantum algorithm and our classical algorithm run in time $\poly(\log N,s)$. (Naturally, exponents in the quantum and classical complexities differ. As shown in the formal statement of Theorem \ref{th:classical-est} (\Cref{th:classical-est-formal}), in the classical case the degree of $P$ appears in the exponent, which is not the case in the quantum case).

 We then show how this result can be used to compute a good estimate of the singular values of a sparse matrix, given a vector that has constant overlap with the relevant subspace.

\begin{theorem}[Informal version of \Cref{th:SVT-formal}]\label{th:SVT}
For any interval $[t_1,t_2]\subseteq [0,1]$, any constant $\theta>0$ and any $s$-sparse matrix $A\in\Comp^{N\times N}$ with $\norm{A}\le 1$ given by query-access, consider the problem of deciding whether
\begin{itemize}
\item[(i)]
$A$ has a singular value in interval $[t_1,t_2]$, or
\item[(ii)]
$A$ has no singular value in interval $(t_1-\theta,t_2+\theta)$.
\end{itemize}
In addition, we are given as input sampling-access to a vector $u\in\Comp^N$, which we are promised in case~$(i)$ satisfies $\norm{\Pi_A^{[t_1,t_2]}u}=\Omega(1)$. Then, there is a $O(\poly(\log N,s))$-time classical algorithm that decides with high probability which of $(i)$ and $(ii)$ holds.
\end{theorem}
\noindent As shown in Section \ref{sub:proofthclassical}, Theorem \ref{th:classical} follows from Theorem \ref{th:SVT}.
Note the condition
$
\norm{\Pi_A^{[t_1,t_2]}u}=\Omega(1)
$
is only required for case (i); in case (ii) there is no condition imposed on $u$.

For context, \emph{quantum} algorithms for the same problem have been considered implicitly in \cite[Section 3.5]{Gilyen+STOC19} and explicitly in \cite[Section 3]{Lin+20} and \cite[Section IV]{Martyn+21}. These quantum algorithms have complexity $O(\poly(\log N,s,1/\theta))$. We thus efficiently dequantize these results for \emph{constant}~$\theta$, in the sense that again both the classical and quantum versions have complexity polynomial in $\log N$ and $s$. %This suffices for us to obtain an essentially tight characterization for the quantum chemistry setting, i.e., $\GLHest$ is classically solvable for \emph{constant} precision $\epsilon$, whereas $\GLHdes$ is $\BQP$-complete for \emph{inverse-polynomial} precision $\abs{b-a}$.

%====================
\subsection{Overview of our techniques}
%===================

\paragraph*{$\BQP$-hardness for inverse-polynomial-precision.} Given a $\BQP$ circuit $U=U_m\cdots\allowbreak U_1$, the natural approach is to apply Kitaev's~\cite{Kitaev+02} circuit-to-Hamiltonian construction to map $U$ to $5$-local Hamiltonian $H=\hin+\hprop+\hout+\hstab$ (a quantum analogue of the Cook-Levin construction~\cite{C72,L73}), so that the ground space of $H$ ``simulates'' $U$ as follows: $\hin$ ensures the computation is initialized correctly at time $t=0$, $\hprop$ that each time step $t>0$ correctly applies $U_{t}$, and $\hout$ penalizes any computation which rejects. Then, one can show~\cite{Kitaev+02} that if $U$ accepts with high (resp., low) probability, then $\lmin{H}$ is small (resp., large). The catch is that $\GLHdes$ \emph{also} requires us to produce a semi-classical state $u$ (as defined in Section~\ref{sub:mainres}) with $\norm{\Pi_{H}u}\ge \delta$. This is the challenge --- in general, one does not know what the \emph{exact} eigenvectors of $H$ are.

To get around this, we apply three tricks. First, in the YES case ($U$ accepts with high probability) it is known that the ``history state'' $\ket{\psihist}$ (\Cref{eqn:psihist} in Section~\ref{sec:Th2}) has ``low energy'' against $H$. By the standard trick of weighting $\hin+\hprop+\hstab$ by a large polynomial $\Delta$, we can force $\ket{\psihist}$ to be $1/\poly(\Delta)$ close to a ground state $\ket{\phi}$. To next ensure $\ket{\phi}$ can be approximated by a semi-classical state, we ``pre-idle\footnote{In contrast, ``idling'' usually refers to the same trick applied at the \emph{end} of $U$, so as to increase the amplitude on time step $m$ of the history state.}'' $U$ by adding a large number $N$ of identity gates to the start of $U$ --- this ensures the first $N$ time steps of $\ket{\psihist}$ have large overlap with a semi-classical state (roughly, the all-zeroes initial state of the computation, coupled with an equal superposition over the first $N$ time steps).

However, this alone does not suffice --- in the NO case ($U$ accepts with low probability), we know nothing about the ground space of $H$ (i.e., we do not even have $\ket{\psihist}$ to guide us). To resolve this, we take inspiration from an arguably unexpected source --- the study of $\pQMAlog$.\footnote{$\pQMAlog$ is the set of decision problems decidable by a poly-time deterministic Turing machine making logarithmically many queries to a $\QMA$ oracle.} There, Ambainis~\cite{A14} introduced a block-encoding technique to embed $\QMA$ query answers into a desired qubit of the ground space. While our work is unrelated to $\pQMAlog$, what this technique allows one to do is encode YES and NO instances of $U$ into orthogonal subspaces, affording us a high degree of control in engineering what the ``NO subspace'' looks like. This, in turn, allows us to ensure the semi-classical state $u$ we design has large overlap with the ground state, even in the NO case.

\paragraph*{Dequantizing the QSVT.} We focus on explaining the main ideas to prove Theorem~\ref{th:classical}, which is a special case of our most general result (Theorem \ref{th:SVT}) that ``dequantizes'' the QSVT. The main idea is best illustrated by considering classical algorithms for the decision version of the guided local Hamiltonian problem, which asks to decide which of $\lambda_H\le a$ and $\lambda_H\ge b$ holds, for $b-a=\Omega(1)$, given sampling-access to a vector such that $\norm{\Pi_H u}\ge \delta=\Omega(1)$. Assume for simplicity that (1) all eigenvalues of $H$ are positive (i.e., $H$ is positive-definite), so that the concepts of singular value decomposition and spectral decomposition coincide (i.e., singular values equal eigenvalues), and that (2) the ground space of $H$ is non-degenerate, i.e., $\lambda_H$ has multiplicity 1.
Then, one has (unknown) spectral decomposition
\[
    H=\sum_{i=1}^{N}\sigma_i u_iu_i^\dagger,
\]
where $0<\sigma_1<\sigma_2\le\sigma_3\le \cdots\le\sigma_N\le 1$, with $\sigma_1=\lambda_H$. Write the vector
$u$ with respect to the same (unknown) basis
$
u=\sum_{i=1}^{N}\alpha_i u_i.
$
From the promise on $u$, we know that $|\alpha_1|\ge \delta=\Omega(1)$.

Conceptually, the QSVT-based approach for this problem works as follows. Consider a smooth-enough function $f\colon[0,1]\to[0,1]$ such as $f(x)\approx 1$ if $x\in[0,a]$ and $f(x)\approx 0$ if $x\in[b,1]$. Applying the function~$f$ to $H$ gives
\[
f(H)=\sum_{i=1}^{N}f(\sigma_i) u_iu_i^\dagger \approx\sum_{ \substack{i\in \{1,\ldots,N\}\\
\sigma_i\le a}}u_iu_i^\dagger,
\]
from which we conclude that we have
\begin{equation}\label{eq:test}
\begin{cases}
\norm{f(H)u}\approx 0 &\textrm{ if } \lambda_H\ge b,\\
\norm{f(H)u}\ge \delta f(\lambda_H)\approx \delta &\textrm{ if } \lambda_H\le a.
\end{cases}
\end{equation}
The objective is to use (\ref{eq:test}) to decide whether $\lambda_H\le a$ or $\lambda_H\ge  b$.
There are now two difficulties: how to choose the threshold function $f$, and how to evaluate the quantity $\norm{f(H)u}$.

To enable efficient evaluation of $f(H)$, the function $f$ is chosen as a polynomial of the smallest possible degree. As shown in prior works on quantum algorithms for problems related to singular value estimation \cite{Low+17, Gilyen+STOC19,Martyn+21}, it turns out that since~$b-a=\Omega(1)$, there exists a constant-degree polynomial that approximates our desired threshold function well.

The key technical issue is to \emph{classically} efficiently evaluate the quantity $\norm{f(H)u}$. For this, we first show how to recursively efficiently compute, $\forall i\in\{1,\ldots,N\}$, the $i$-th entry of the vector $f(H)u$ in $\poly(\log N,s)$ time given query-access to the vector $u$ (which we have). This crucially uses the facts that $f$ is a \emph{constant}-degree-$d$ polynomial and $H$ an \emph{$s$-sparse} matrix. Then, the key technical insight is that given sampling-access to a vector $v$, it is possible to compute a good approximation of $v^\dagger f(H)u$ in $\poly(\log N,s)$ time as well. To prove this result, we combine techniques from prior works on dequantization \cite{TangSTOC19} with estimation techniques tailored to our application (this is the main contribution of Section \ref{sec:main}).
Roughly, the approach is natural --- to estimate an inner product $v^\dagger f(H)u$ of exponentially long vectors, we exploit our sampling assumption to ``pick out'' the ``heaviest/most important'' terms in the (sum comprising the) inner product. By rescaling and averaging these appropriately, one can show via the Chebyshev inequality that $\poly(s^d, 1/\epsilon)$ random samples suffice to obtain $\epsilon$ precision.
Finally, taking $v=u$, which is possible since we have sampling-access to~$u$, enables us to calculate an approximation of $u^\dagger f(H)u$, which leads to an approximation of $\norm{f(H) u}$, and then enables us to use (\ref{eq:test}) to decide which of $\lambda_H\le a$ and $\lambda_H\ge  b$ holds.
%=====================================================
\subsection{Implications to quantum chemistry and the quantum PCP conjecture}\label{sub:app}
%=====================================================
We now describe applications of our dequantization technique to quantum chemistry and to the quantum PCP conjecture.

%===
\paragraph*{Towards a quantum advantage for quantum chemistry.}
%===
The problem $\GLHest$ may be viewed as a theoretical abstraction of the core problem considered in quantum algorithmic proposals for quantum chemistry based on quantum phase estimation: computing an estimate of the ground state energy of a local Hamiltonian, given access to a guiding state sufficiently close to the ground state. Indeed, essentially all those proposals use as guiding state a quantum state that has an efficient classical description that enables sampling-access (e.g., the state obtained by the Hartree-Fock method). With respect to this viewpoint, Theorem \ref{th:hardness} shows that computing an \emph{inverse-polynomial} estimate, which is the level of precision needed for quantum chemistry (\Cref{sub:motivations}), is $\BQP$-hard.\footnote{Additionally, as explained in Section \ref{sub:mainres}, when a quantum state representing the guiding state can be constructed efficiently (e.g., when the guiding state is a semi-classical state), this problem can be solved in quantum polynomial-time.}
%\footnote{More precisely, Theorem \ref{th:hardness} shows that $\GLHest$ with inverse-polynomial precision is $\BQP$-hard. As discussed in Section~\ref{sub:mainres}, containment of $\GLHest$ in $\BQP$ follows from existing work~\cite{Martyn+21}.% The problem can be solved in quantum polynomial-time when a quantum state representing $u$ can be constructed efficiently, as explained in Section \ref{sub:mainres}.
%\snote{i think we get containment in BQP, no, at least for the formal statement of $\GLHest$? I updated the text of this footnote based on this.}}
%Theorem \ref{th:classical}, on the other hand, shows that endowing a classical computer with the now-standard assumption of \emph{sampling access} \cite{TangSTOC19} to the guiding state (i.e., $\GLHest$), yields a perhaps unexpected result: The ground state energy can also be efficiently classically estimated, but only up to \emph{constant} precision.
Theorem \ref{th:classical}, on the other hand, shows that this problem can be solved efficiently with constant precision on a classical computer.
This gives theoretical evidence that the level of precision required for quantum chemistry problems can only be achieved on \emph{quantum} computers, not classical ones. Thus, quantum chemistry may indeed be a practical application at which quantum computers demonstrate ``superiority'' over classical ones.

Of course, there are two important caveats to keep in mind: (1) By definition, $\GLHest$ assumes possession of a good ``guiding state'' $u$. As discussed in Section \ref{sub:motivations}, although in \emph{practice} such a guiding state can often be found, in the \emph{worst} case, even the existence of such a state is not generally guaranteed. This is because $\GLHest$ without $u$ is essentially the $k$-local Hamiltonian problem (defined below), which is $\QMA$-complete~\cite{Kitaev+02} --- and whether ground states of $\QMA$-hard local Hamiltonians have efficient quantum preparation circuits remains a major open question (i.e., is $\QCMA=\QMA$?). (2) In practice, the particular properties of a given local Hamiltonian (such as geometric constraints) could in principle be exploited to bypass the worst-case $\BQP$-hardness result of Theorem \ref{th:hardness} in special cases. Indeed, one of the main appeals of recent ``quantum supremacy/advantage'' frameworks such as Boson Sampling~\cite{AA11} (or Random Circuit Sampling~\cite{AAM19,BFNV19}) is that, in the long-term setting where ``noise-free'' quantum computers may become a reality (i.e., the setting in which full-scale quantum phase estimation for quantum chemistry is feasible), one can rigorously prove~\cite{AA11} not just \emph{worst}-case hardness of Boson Sampling for classical computers, but even \emph{average}-case hardness! Whether a similar average-case hardness result could hold for $\GLHest$ is left as an open question. Nevertheless, the appeal of $\GLHest$ (versus, say, Boson Sampling) is that it is a \emph{practically motivated} task, whereas all existing near-term ``quantum advantage'' proposals (such as Boson Sampling and Random Circuit Sampling) are primarily ``proofs-of-principle'', i.e., do not have strong practical applications.

Finally, note that our classical simulation (\Cref{th:classical}) applies to general sparse Hamiltonians.\footnote{While the statement of \Cref{th:classical} given in Section \ref{sub:mainres} is for $k$-local Hamiltonians with $k=O(\log n)$, the same classical simulation result holds more generally for $s$-sparse Hamiltonians with $s=O(\poly(n))$, since the dequantization results of Section~\ref{sec:SVT} apply to such sparse matrices as well.} Thus, if one \emph{did} have a good guiding state for a particular fermionic Hamiltonian $H_f$ from quantum chemistry, one could in principle apply a fermion-to-qubit mapping~\cite{JW28,BK02,VC05,WHVT16,SW19,SBMW19,JMBN19,DKBC21} to obtain a sparse Hamiltonian $H$, and then apply \Cref{th:classical}.
%=====================================================
%\subsection{Applications to the Quantum PCP Conjecture}\label{sec:PCP}
%=====================================================
%=====================================================
 %We first describe the quantum PCP conjecture, and then describe the implications of our results.
%===
\paragraph*{Quantum PCP conjecture.}\label{sub:PCP-QMA}
%===
The standard (i.e., non-guided) Local Hamiltonian problem considered in quantum complexity theory, which we denote by $\LH(k,a,b)$, where $k\ge2$ is an integer and $a,b\in[-1,1]$ are two real numbers such that $a<b$, is defined as follows (see, e.g., \cite{Gharibian+15}).
\begin{center}
\fbox{
\begin{minipage}{15 cm} \vspace{2mm}

\noindent$\LH(k,a,b)$\hspace{3mm}($k$-local Hamiltonian problem)\\\vspace{-3mm}

\noindent\hspace{3mm} Input: a $k$-local Hamiltonian $H$ acting on $n$ qubits such that $\norm{H}\le 1$

\noindent\hspace{3mm}
Promise: $\lambda_{H}\le a$, or $\lambda_{H}\ge b$
\vspace{2mm}

\noindent\hspace{3mm} Goal: decide which of $\lambda_{H} \le a$ or $\lambda_{H}\ge b$ holds
\vspace{2mm}
\end{minipage}
}
\end{center}
\noindent The standard version of the quantum PCP conjecture (see, e.g., \cite{Aharonov+13,Grilo18}) states that this problem is $\QMA$-hard even for $k=O(1)$ and $b-a=\Omega(1)$. (Note the norm constraint $\norm{H}\leq 1$ in $\LH$ is crucial for this to be a meaningful statement, as otherwise one could simply multiply~$H$ by a large enough polynomial to achieve $b-a\in\Omega(1)$.)
\begin{conjecture}[quantum PCP conjecture]\label{conj:PCP}\label{conj:PCP-QMA}
There exist $k=O(1)$ and $a,b\in[-1,1]$ satisfying $b-a=\Omega(1)$ such that $\LH(k,a,b)$ is $\QMA$-hard.
\end{conjecture}
\noindent This conjecture is one of the central conjectures in quantum complexity theory.

In Section \ref{sec:PCP}, we introduce the concept of quantum states (i.e., unit complex vectors) having a succinct representation that allows sampling-access (Definition~\ref{def:description}).\footnote{We note that a related approach has been used by Bravyi \cite{Bravyi2015} and Liu \cite{Liu2021}. The latter introduced the notion of an ``easy witness'' for stoquastic local Hamiltonian, and showed that the Stoquastic Local Hamiltonian instance with an easy witness in the YES case can be verified in $\MA$. The approach from \cite{Liu2021} considers query access and sample access to the witness. This approach actually works for inverse-polynomial precision, whereas ours (Theorem \ref{th:qPCP}) works only for constant precision.} This concept is natural, in view of the correspondence between quantum states and sampling-access to vectors that has now become standard, and is used both in the present work and in prior works on dequantization \cite{Chia+STOC20,Chia+20,Gilyen+20,Jethwani+MFCS20,TangSTOC19}. A consequence of our dequantization result is the following result (a formal version is stated and proved in Section  \ref{sec:PCP}).
\begin{theorem}[Informal version of \Cref{th:qPCP-formal}]\label{th:qPCP}
Consider problem $\LH(k,a,b)$ with the following additional promise: the ground state of $H$ has constant overlap with a complex vector that has a succinct representation allowing sampling-access.
For any $k=O(\log n)$, any $a,b\in[-1,1]$ such that $b-a=\Omega(1)$, this variant of $\LH(k,a,b)$ is in the complexity class $\MA$.
\end{theorem}
We stress that Theorem \ref{th:qPCP} crucially requires $b-a=\Omega(1)$. Under the widely believed assumption $\QMA\neq \MA$, this theorem can thus be interpreted either as evidence that the quantum PCP conjecture (Conjecture \ref{conj:PCP}) is false, or as evidence that the ground states of the Hamiltonians involved in the quantum PCP conjecture do not have good approximations as states with succinct representations allowing sampling-access.

To discuss the latter possibility further, recall that assuming $\QCMA\neq\QMA$, the only way the quantum PCP conjecture could hold is if ground states of $\QMA$-hard Hamiltonians cannot be reasonably approximated by (even non-uniform) poly-size quantum circuits. Thus, a first step towards proving the quantum PCP conjecture would be to argue that {local} Hamiltonians with such ``complex'' low-energy spaces even exist in the first place. This is essentially the premise of Freedman and Hasting's No Low-Energy Trivial States (NLTS) conjecture\footnote{We remark that since the short version of the article was published~\cite{GG22}, Anshu, Breuckmann, and Nirkhe have proved the NLTS conjecture \cite{ABN22}. For clarity, our NLSS conjecture remains open.} \cite{Freedman+14,Eldar+FOCS17}. To state the latter, one defines the concept of trivial states: a family of quantum states $\{\ket{\varphi_n}\}_{n\in \Nat}$ is \emph{trivial} if the states in the family have low-complexity entanglement, in the sense that they can be generated by a family of constant-depth quantum circuits.
%\footnote{The original conjecture considers families of mixed states (instead of  pure states).}

\begin{conjecture}[NLTS conjecture]
There exist a family of $O(1)$-local Hamiltonians $\{H_n\}_{n\in\Nat}$, where each $H_n$ acts on $n$ qubits, and a constant $\epsilon>0$ such that for any family of trivial states $\{\ket{\varphi_n}\}_{n\in \Nat}$, where each $\ket{\varphi_n}$ is an $n$-qubit state, we have for any sufficiently large~$n$:
\[
\bra{\varphi_n}H_n\ket{\varphi_n}>\lambda_{H_n}+\epsilon.
\]
\end{conjecture}

\noindent Note that NLTS is based on the presumption that $\NP\neq \QMA$ (since expectation values for local observables against states preparable by constant-depth circuits can be classically efficiently evaluated via a light-cone argument\footnote{We thank Matthew Coudron for pointing this out to us.}), Theorem \ref{th:qPCP} is based on the presumption $\MA\neq\QMA$. In words, if the quantum PCP conjecture holds and $\MA\neq\QMA$, Theorem \ref{th:qPCP} implies that, roughly, sampling access to any low-energy state of a $\QMA$-hard local Hamiltonian should not be possible. We formalize this via the following variant of the NLTS conjecture, which we call the No Low-Energy Samplable States (NLSS) conjecture.

\begin{conjecture}[NLSS conjecture --- informal version of \Cref{conj:variant-formal}]\label{conj:variant}
There exist a family of $O(1)$-local Hamiltonians $\{H_n\}_{n\in\Nat}$, where each $H_n$ acts on $n$ qubits, and a constant $\epsilon>0$ such that for any family of quantum states $\{\ket{\varphi_n}\}_{n\in \Nat}$ that has a succinct representation allowing sampling-access, where each $\ket{\varphi_n}$ is an $n$-qubit state, we have for any sufficiently large~$n$:
\[
\bra{\varphi_n}H_n\ket{\varphi_n} >\lambda_{H_n}+\epsilon.
\]
\end{conjecture}
\noindent (The formal statement of Conjecture \ref{conj:variant} is given in Section \ref{sec:PCP}.) In principle, since the assumption $\MA\neq \QMA$ is stronger than $\NP\neq\QMA$, the NLSS conjecture ``should'' be ``easier'' to prove than NLTS. Conversely, if the NLSS conjecture can be shown false, then the quantum PCP conjecture is false. Morally, this raises the interesting possibility that the ``straw which breaks the quantum PCP camel's back'' is the same as that used to refute quantum advantages for (say) recommender systems and certain quantum machine learning algorithms --- Tang's classical sampling-access assumption~\cite{TangSTOC19}.

%=====================================
\section{Preliminaries}\label{sec:prelim}
%=====================================
%For any $n\ge 1$, let us write $\Hh_n=\Comp^{2^n}$. We typically write vectors in $\Hh_n$ that have norm 1 (i.e., quantum states) using the ket notation, and use symbols like $u$ or $v$ to represent arbitrary vectors in $\Hh_n$.
In this paper we assume that arithmetic operations in $\Comp$ (i.e., addition, subtraction, multiplication and division of two scalars) can be done at unit cost.\footnote{All our upper bounds can easily be converted into upper bounds for the bit-complexity of the problem as well if the notions of query-access and sampling-access to the input are redefined appropriately to take in consideration the bit-complexity of the input vectors and matrices.} For any integer $s\ge 1$, we say that a matrix is $s$-sparse if each row and column has at most $s$ non-zero entries. We say that a polynomial $P\in\Real[x]$ is even if $P(x)=P(-x)$ for all $x\in \Real$, and say that $P$ is odd if $P(x)=-P(-x)$ for all $x\in \Real$.
% (i.e., a polynomial $p\in\Real[x]$ such that $p(x)=p(-x)$)
%We use $M$ and $N$ to write the dimension of vectors and matrices. We always implicitly that $M$ and $N$ are polynomially related.

For $\epsilon\in(0,1)$, we say that a real number $\tilde \mu$ is an $\epsilon$-additive-approximation of a real number~$\mu$ if $|\tilde \mu -\mu|\le \epsilon$.
As in \cite{Montanaro+05}, we will use the following ``powering lemma'' from \cite{Jerrum+86} to amplify the success probability of probabilistic estimators.\footnote{Lemma 6.1 in \cite{Jerrum+86} is actually stated in terms of multiplicative approximation (approximation within an $1+\epsilon$ ratio), rather than additive approximation. The proof, however, only uses Chernoff's bound and thus applies to additive approximation as well.}
\begin{lemma}[Lemma 6.1 in \cite{Jerrum+86}]\label{lemma:powering}
Let $\Aa$ be an algorithm that produces with probability at least $3/4$ an $\epsilon$-additive-approximation of a quantity $\mu$. Then, for any $\delta>0$, it suffices to repeat $\Aa$ $O(\log(1/\delta))$ times and take the median to obtain an $\epsilon$-additive-approximation of $\mu$ with probability at least $1-\delta$.
\end{lemma}

%Given two probability distributions $p,q\colon\{1,\ldots,N\}\to[0,1]$ we denote
%\[
%\stat{p-q}=\frac{1}{2}\sum_{i=1}^{N}|p(i)-q(i)|
%\]
%their statistical distance.
%===
\subsection{Query-access and sampling-access to vectors and matrices}
%===

We first define two natural ways of accessing vectors and matrices. First, we define query-access to a vector, which corresponds to the standard model where the input is stored in a (classical) random-access-memory with polylogarithmic-time access cost.
\begin{definition}\label{def:query-access}
%We say that we have query-access to a matrix $A\in \Comp^{M\times N}$ if we can make the following query in $O(\log M+\log N)$ time:
%given two indices $(i,j)\in\{1,\ldots,M\}\times \{1,\ldots,N\}$, output the entry $A_{ij}$.
We say that we have query-access to a vector $u\in \Comp^{N}$ if there exists an $O(\poly(\log N))$-time classical algorithm $\Qq_u$ that
%we can make the following query in $O(\log N)$ time:
on input $i\in\{1,\ldots,N\}$, outputs the entry $u_i$.
\end{definition}

We define query-access to sparse matrices in a way that allows to efficiently recover the non-zero entries.
\begin{definition}\label{def:sparsequery}
We say that we have query-access to an $s$-sparse matrix $A\in \Comp^{M\times N}$ if we have access to two $O(\poly(\log MN))$-time classical algorithms $\Qq_A^{\mathrm{row}}$ and $\Qq_A^{\mathrm{col}}$ such that
\begin{itemize}
\item
on input $(i,\ell)\in\{1,\ldots,M\}\times\{1,\ldots,s\}$, the algorithm $\Qq_A^{\mathrm{row}}$ outputs the $\ell$-th non-zero entry of the $i$-th row of $A$ if this row has at least $\ell$ non-zero entries, and outputs an error message otherwise;
\item
on input $(j,\ell)\in\{1,\ldots,N\}\times\{1,\ldots,s\}$, the algorithm $\Qq_A^{\mathrm{col}}$ outputs the $\ell$-th non-zero entry of the $j$-th column of $A$ if this column has at least $\ell$ non-zero entries, and outputs an error message otherwise.
\end{itemize}
\end{definition}

We now define the concept of sampling-access to a vector. Compared with prior works \cite{ Chia+STOC20,Chia+20,Gilyen+20,Jethwani+MFCS20,TangSTOC19}, our definition introduces a slight generalization that we call $\zeta$-sampling-access, where $\zeta\in[0,1)$ is a parameter (for $\zeta=0$ we obtain the same notion of sampling-access as in prior works). Note that \cite{Chia+STOC20} also considered a slighly different generalization, which was called $\phi$-oversampling.
% Our definition actually even allows some approximate sampling-access.

\begin{definition}\label{def:sample}
For a parameter $\zeta\in[0,1)$, we say that we have $\zeta$-sampling-access to a vector $u\in \Comp^{N}$ if the following three conditions are satisfied:
\begin{itemize}
\item[(i)]
we have access to an $O(\poly(\log N))$-time classical algorithm $\Qq_u$ that on input $i\in\{1,\ldots,N\}$, outputs the entry $u_i$;
\item[(ii)]
we have access to an $O(\poly(\log N))$-time classical algorithm $\Sq_u$ that samples from a probability distribution $p\colon\{1,\ldots,N\}\to [0,1]$ such that
\[
p(j)\in\left[(1-\zeta)\frac{|u_j|^2}{\norm{u}^2},(1+\zeta)\frac{|u_j|^2}{\norm{u}^2}\right]
\]
for all $j\in\{1,\ldots,N\}$;
%\[
%\stat{p-q}\le \epsilon
%\]
%where $q\colon\{1,\ldots,N\}\to [0,1]$ is the distribution defined as $p(j)=|u_j|^2/\norm{u}^2$ for all $\{1,\ldots,N\}$.
\item[(iii)]
we are given a real number $m$ satisfying $|m-\norm{u}|\le \zeta \norm{u}$.
\end{itemize}
%We call the triple $(\rep(\Qq_u),\rep(\Sq_u),m)$ a classical description of $u$.
We simply say that we have sampling-access to $u$ (without specifying $\zeta$) if we have $0$-sampling-access.
\end{definition}

%Let $\Sample_n$ denote the set of all vectors $u\in \Hh_n$ that have a classical $\poly(n)$-bit representation $\mathsf{rep}(u)$ that enables to perform sampling access in classical polynomial time (i.e., each of the three queries of Definition \ref{def:sample} can be performed in $\poly(n)$ time given $\mathsf{rep}(u)$). We call $\mathsf{rep}(u)$ the \emph{classical representation} of $u$.

%===
\subsection{The Singular Value Transformation}
%===
Any matrix $A\in \Comp^{M\times N}$ can be written in the form
\begin{equation}\label{eq:SVT2}
A=\sum_{i=1}^{\min(M,N)}\sigma_i u_iv_i^\dagger,
\end{equation}
where the $\sigma_i$'s are nonnegative real numbers, $\{u_i\}_i$ are orthonormal vectors in $\Comp^M$ and $\{v_i\}_i$ are orthonormal vectors in $\Comp^N$. Such a decomposition is called a Singular Value Decomposition of~$A$.  Note that while the $\sigma_i$'s are unique (up to reordering), the choice of the $u_i$'s and $v_i$'s is not unique. When $M\ge N$ we obviously have $\min(M,N)=N$. When $M<N$, it will be convenient to also write the decomposition as
\begin{equation}\label{eq:SVT}
A=\sum_{i=1}^{N}\sigma_i u_iv_i^\dagger
\end{equation}
by using the following convention:
for each $i\in\{M+1,\ldots, N\}$ we set $\sigma_i=0$ and choose an arbitrary vector $u_i$ in $\Comp^M$, and we choose the $v_i$'s for $i\in\{M+1,\ldots, N\}$ so that $\{v_i\}_{i\in\{1,\ldots,N\}}$ forms an orthonormal basis~$\Comp^N$.
%Note that when $s\ge t$ the two versions $(\ref{eq:SVT2})$ and $(\ref{eq:SVT})$ coincide.

Consider an even polynomial $P\in\Real[x]$ and write
\[
P(x)=a_0+a_2x^2+a_4x^4+\cdots+ a_{2d}x^{2d},
\]
where $2d$ is the degree of $P$ and $a_0,\ldots,a_{2d}$ are real coefficients. The polynomial singular value transformation of $A$ associated to $p$ is defined as
\[
P(\sqrt{A^\dagger A})=a_0I_N+a_2A^\dagger A+a_4(A^\dagger A)^2+\cdots+a_{2d}(A^\dagger A)^d,
\]
where $I_N$ denotes the identity matrix of dimension $N$.
Using Eq.~(\ref{eq:SVT}), we obtain
\begin{align*}
P(\sqrt{A^\dagger A})&=a_0I_t+a_2\sum_{i=1}^{N}\sigma_i^2 v_iv_i^\dagger+a_4\sum_{i=1}^{N}\sigma_i^4 v_iv_i^\dagger+\cdots+a_{2d}\sum_{i=1}^{N}\sigma_i^{2d} v_iv_i^\dagger\\
&=\sum_{i=1}^{N}P(\sigma_i) v_iv_i^\dagger.
\end{align*}

For any $a,b\in [0,1]$ such that $a<b$, we define
\[
\Pi_A^{[a,b]}=\sum_{\substack{i=\{1,\ldots,N\}\\
\sigma_i\in [a,b]}} v_iv_i^\dagger.
\]
Note that $\Pi_A^{[a,b]}$ is the orthogonal projector into the subspace of the row space corresponding to the singular values in the interval $[a,b]$.

\begin{remark} The Singular Value Transformation can also be defined for an odd polynomial $P\in\Real[x]$. In this case, however, the definition needs to be slightly adjusted by swapping the left singular vectors $\{u_i\}_i$ and the right singular vectors~$\{v_i\}_i$. Our dequantization technique can easily be adapted to work in this case as well, but to avoid confusion in this paper we discuss only the case of even polynomials, which is enough for our application to estimating the singular values of a sparse matrix.
\end{remark}
%=====================================
\section[BQP-Hardness of Polynomial-Precision Eigenvalue Estimation]{$\BQP$-Hardness of Polynomial-Precision Eigenvalue Estimation}\label{sec:Th2}
%=====================================
%\input{hardness}
In this section we prove the $\BQP$-Hardness of the guided local Hamiltonian problem with inverse-polynomial precision (Theorem \ref{th:hardness}).

\begin{proof}[Proof of \Cref{th:hardness}]
    Let $\Pi=(\Piyes,\Pino)$ be a promise problem in $\BQP$, and $x\in\set{0,1}^n$ an input. Let $U=U_m\cdots U_1$ be the poly-time uniform circuit consisting of $1$- and $2$-qubit gates $U_i$ deciding $\Pi$. Specifically, $U$ takes in an $n$-bit input register $A$, and $p(n)$-qubit (for polynomial~$p$) ancilla register $B$ initialized to $\ket{0}^{\otimes p(n)}$. If $x\in\Piyes$ (resp., $x\in\Pino$), then measuring $U$'s output qubit in the standard basis yields $1$ with probability at least $\alpha$ (resp.,~at most $\beta$). Via standard error reduction for $\BQP$ (via parallel repetition), we may assume $\alpha=1-2^{-n}$ and $\beta=2^{-n}$,
% for $\abs{\alpha-\beta}\geq 1/\poly(n).$
We give a poly-time many-one reduction from $(x,U)$ to an instance of $H$ of $\GLHdes(k,a,b,\delta)$, for $k,a,b,\delta$ to be chosen as needed.

    \paragraph*{Kitaev's construction.} To begin, for completeness we restate Kitaev's~\cite{Kitaev+02} $5$-local circuit-to-Hamiltonian construction elements applied to our $U$:
\begin{eqnarray*}
    \hin&:=&
    \left(I-\ketbra{x}{x}\right)_{A}\otimes\left(I-\ketbra{0\cdots 0}{0\cdots 0}\right)_{B}\otimes \ketbra{0}{0}_{C}\label{eqn:hin}\\
    \hout&:=&\ketbra{0}{0}_{\textup{out}}\otimes
     \ketbra{m}{m}_{C}\label{eqn:hout}\\
    \hstab&:=&\sum_{j=1}^{m-1}\ketbra{0}{0}_{C_j}\otimes\ketbra{1}{1}_{C_{j+1}}\label{eqn:hstab}\\
    \hprop &:=& \sum_{t=1}^{m} H_t {\rm,~where~ }H_t{\rm ~is ~defined~ as}\label{eqn:hprop}\\
    H_t&:=&-\frac{1}{2}U_t\otimes\ketbra{t}{ {t-1}}_{C} -\frac{1}{2}U_t^\dagger\otimes\ketbra{{t-1}}{t}_{C} +\frac{1}{2}I\otimes(\ketbra{ t}{t}+\ketbra{ {t-1}}{ {t-1}})_{C},\label{eqn:ht}
\end{eqnarray*}
 where:
\begin{itemize}
\item
$A$, $B$, and $C$ are the input, ancilla, and clock registers, respectively,
\item
$C$ is encoded implicitly in unary,
\item
$C_j$ denotes the $j$th qubit of register $C$,
\item
$\ketbra{0}{0}_{\textup{out}}$ projects onto the dedicated output wire of $U$ (for example, the first qubit of $B$),\footnote{In contrast to \cite{Kitaev+02}, which studied $\QMA$, here we do not require a separate proof register, as $\BQP$ does not take in a proof.}
\item
finally, all  Hamiltonians above are positive semi-definite.
\end{itemize}
The circuit-to-Hamiltonian construction $H'=\hin+\hprop+\hout+\hstab$ has the following property.

 \begin{fact}[\cite{Kitaev+02}]\label{fact:1}
    Let $H'=\hin+\hprop+\hout+\hstab$. Then,
 \begin{itemize}
    \item if $U$ accepts $x$ with probability at least $\alpha$, then $\lmin{H'}\leq \frac{1-\alpha}{m+1}=:\alpha'$,
    \item if $U$ accepts $x$ with probability at most $\beta$, then $\lmin{H'}\geq \Omega\left(\frac{1-\sqrt{\beta}}{m^3}\right)=:\beta'$.
 \end{itemize}
 \end{fact}
We state a helpful lemma lower bounding the non-zero eigenvalues on the frustration-free component of this construction.
\begin{lemma}[Lemma 2.2 of \cite{Gharibian+19} (based on\footnote{Alternatively, see Lemma 23 of the arXiv version \cite{GK12_2}.} Lemma 3 of \cite{GK12})]\label{l:GKgap}
    The smallest non-zero eigenvalue of $\hin+\hprop+\hstab$ is at least $\pi^2/(64m^3)$
%\in\Omega(1/m^3)$
for $m\geq 1$.
\end{lemma}
 Since $\alpha=1-2^{-n}$ and $\beta=2^{-n}$, we have that estimating $\lmin{H'}$ to inverse-polynomial precision suffices to decide $x$. In the YES case, the upper bound $\alpha'$ will be attained via the history state, which we now state:
 \begin{equation}\label{eqn:psihist}
    \ket{\psihist}=\frac{1}{\sqrt{m+1}}\sum_{t=0}^mU_t\cdots U_1\ket{x}_{A}\ket{0\cdots 0}_{B}\ket{t}_{C}.
\end{equation}

 \paragraph*{Our construction.} The rub is that $\GLHdes$ also requires us to produce, as part of the reduction, a semi-classical state $\ket{u}$ (as defined in Section \ref{sub:mainres}) with $\norm{\Pi_{H'}\ket{u}}\ge \delta$. Unfortunately, in the YES case, if $\alpha\neq 1$ then the history state $\ket{\psihist}$ is not necessarily an actual ground state of $H'$. Even worse, the NO case analysis goes via Kitaev's Geometric Lemma~\cite{Kitaev+02}, which says nothing about the ground state of $H$ (i.e., we do not even have the history state to guide us).

 To get around this, we apply three tricks. First, we apply a standard trick of weighing $\hin+\hprop+\hstab$ by a large polynomial pre-factor $\Delta$, to allow us to argue about the structure of the low energy space of $H'$. Second, we ``pre-idle'' the verifier, by which we mean we update $U$ to a $M:=(m+N)$-gate circuit, where the first $N$ gates are $I$, and the last $m$ gates are the original circuit $U$. (In contrast, ``idling'' usually refers to the same trick applied at the \emph{end} of $U$, so as to increase the amplitude on time step $m$ of the history state. Here, we increase the weight on step~$0$.) At this point, we have
 \begin{equation}
    H':=\Delta\hin+\Delta\hprop+\Delta\hstab+\hout
 \end{equation}
 for $\hprop$ using our new $M$-step $U$, and where we update our definitions of $\alpha'$ and $\beta'$ to
\[
\alpha':=\frac{1-\alpha}{M+1} \hspace{3mm}\textrm{ and } \hspace{3mm}\beta':=\Omega\left(\frac{1-\sqrt{\beta}}{M^3}\right).
\]
Third, we take inspiration from Ambainis' $\QMA$ query gadget construction~\cite{A14}, and define:
 \begin{eqnarray*}
    H &:=& \frac{\alpha'+\beta'}{2}I_{ABC}\otimes\ketbra{0}{0}_D + H'_{ABC}\otimes \ketbra{1}{1}_D,\\
    \ket{u} &:=& \ket{x}_A\ket{0\cdots 0}_{B}\left(\frac{1}{\sqrt{N}}\sum_{t=1}^{N}\ket{t}\right)_{C}\ket{+}_{D}.
 \end{eqnarray*}
 Since we will shortly set $N\in\poly(n)$, $\ket{u}$ is semi-classical (i.e., an equally weighted superposition of $\poly(n)$ strings), as desired. We remark that a similar idea of using block-encoding relative to a single qubit register (in our context, $D$) to force $H$ to encode $\QMA$ query answers in $D$ was used by~\cite{A14}, albeit in the unrelated context of class $\pQMAlog$. While our aim is not to study $\pQMAlog$, this type of gadget nevertheless will allow us to precisely pin down the ground space of $H$ in the NO case. (For clarity, $\ket{u}$ is specific to this paper, i.e., the concept of $\ket{u}$ is not relevant to~\cite{A14}.) Note $\ket{u}$ sums up to time step $N$, \emph{not} $M$, and starts the summation at $t=1$.

 Finally, set $k=6$, $a=\alpha'$, $b=(\alpha'+\beta')/2$, and
any $\delta\in (0,1/\sqrt{2}-\Omega(1/\poly(n)))$. Set $N$ as a large enough fixed polynomial so that $N/(2(M+1))\geq 1/4+\Delta'$ for some $\Delta'$ defined later ($\Delta'$ will scale as $1/4-1/\poly(n)$). We will also set $\Delta>0$ as needed. This completes the construction, which clearly runs in $\poly(n)$ time. As an aside, as defined, $H$ will have $\poly(n)$ norm, whereas $\GLHdes$ asks for $\norm{H}\leq 1$. This is easily corrected in the standard manner by dividing $H$ by any polynomial bound on $\norm{H}$ (obtained, e.g., by applying the triangle inequality for $\norm{\cdot}$ over the local terms of $H$). This scales down the thresholds $\alpha',\beta'$ (which is fine, since they remain inverse poly gapped), and leaves the eigenvectors invariant, which is important for our claims regarding $\norm{\Pi_{H}\ket{u}}\ge \delta$.

 \paragraph*{Correctness.} For clarity, recall $U$ is now a $M$-gate circuit, with $M=m+N$. \\

 \noindent\emph{YES case.} Suppose $x\in \Piyes$.
 We must show that $\lmin{H}\leq a$ (to satisfy the YES case requirement in $\GLHdes$) and $\norm{\Pi_{H}\ket{u}}\ge \delta$ (to satisfy the guiding state promise in $\GLHdes$).
 For the first claim, since $x\in \Piyes$, $U$ accepts $x$ with probability $\alpha$.
 Thus, we may upper bound $\lmin{H}$ using assignment $\ket{\psihist}_{ABC}\ket{1}_D$ as follows:
 \begin{eqnarray}
    \lmin{H}&\leq& \bra{\psihist}_{ABC}\bra{1}_DH\ket{\psihist}_{ABC}\ket{1}_D\nonumber\\
    &=&\bra{\psihist}_{ABC}H'\ket{\psihist}_{ABC}\nonumber\\
    &\leq&(1-\alpha)/(M+1)\nonumber\\
    &=&\alpha'\label{eqn:overlap},
 \end{eqnarray}
 where the second statement follows since $\hin\ket{\psihist}=\hprop\ket{\psihist}=\hstab\ket{\psihist}=0$, and the third by direct calculation (in the process re-deriving the first point of \Cref{fact:1}). (Recall we updated $m$ in Kitaev's original construction at the start of this discussion to $M$.) Thus, in this case $\lmin{H}=\lmin{H'}\leq a$, as desired.

 To next show $\norm{\Pi_{H}\ket{u}}\ge \delta$, observe that since we just showed $\lmin{H}\leq\alpha'$ (Equation (\ref{eqn:overlap})), and since $\beta'>\alpha'$, we have that $\lmin{H}$ lies in the $1$-block relative to $D$, i.e., $\lmin{H}$ is achieved only by an eigenvector of the form\footnote{In other words, $H$ is, up to permutation, block diagonal with respect to the standard basis on $D$, and $\lmin{H}$ lies in the $1$-block of this decomposition.} $\ket{\phi}_{ABC}\ket{1}_D$ for some $\ket{\phi}_{ABC}$.
 We claim that any such $\ket{\phi}_{ABC}$ must have large overlap on the history state $\ket{\psihist}_{ABC}$, which will in turn allow us to conclude that $\ket{\phi}_{ABC}$ has large overlap with $\ket{u}$. Since $\ket{\phi}_{ABC}\ket{1}_D$ was a ground state of $H$, this will allow us to conclude that $\norm{\Pi_{H}\ket{u}}\ge \delta$, i.e., that $\ket{u}$ is a valid guiding state.

  To begin, letting $\spa{N}$ denote the null space $\textup{Null}(\hin+\hprop+\hstab)$, recall that $\spa{N}$ is $1$-dimensional and spanned by $\ket{\psihist}$.\footnote{That $\spa{N}$ is spanned by history states is well-known from Kitaev's original analysis~\cite{Kitaev+02}; the $1$-dimensionality arises since there is no proof register in our setting, i.e., we are reducing from $\BQP$ rather than $\QMA$.} Let $\ket{\phi}$ be an arbitrary ground state of $H'$, implying $\ket{\phi}_{ABC}\ket{1}_D$ is a ground state of $H$. Write $\ket{\phi}=\eta_1\ket{\psihist}+\eta_2\ket{\psihist^\perp}$ for $\ket{\psihist}\in \spa{N}$, $\ket{\psihist^\perp}\in \spa{N^\perp}$, and $\abs{\eta_1}^2+\abs{\eta_2}^2=1$. Then, since $H'\succeq 0$ (since recall $\hin,\hprop,\hout,\hstab\succeq 0$ and $\Delta\geq 0$),
 \begin{eqnarray*}
    \alpha'&\geq& \bra{\phi}H'\ket{\phi}\\
    &\geq& \bra{\phi}\Delta(\hin+\hprop+\hstab)\ket{\phi}\\
    &=&\abs{\eta_1}^2\bra{\psihist}\Delta(\hin+\hprop+\hstab)\ket{\psihist}+\\
    &&\abs{\eta_2}^2\bra{\psihist^\perp}\Delta(\hin+\hprop+\hstab)\ket{\psihist^\perp}\\
    &=& \abs{\eta_2}^2\bra{\psihist^\perp}\Delta(\hin+\hprop+\hstab)\ket{\psihist^\perp}\\
    &\geq& \abs{\eta_2}^2\frac{c\Delta}{M^3}
 \end{eqnarray*}
for $c=\pi^2/64=\Omega(1)$, where the second statement holds since $\hout\succeq 0$ (and thus $H'\succeq \Delta(\hin+\hprop+\hstab)$), the third and fourth since $(\hin+\hprop+\hstab)\ket{\psihist}=0$, and the last by \Cref{l:GKgap}. We conclude that
\begin{equation}\label{eqn:step1}
    \abs{(\bra{\phi}_{ABC}\bra{1}_D)(\ket{\psihist}_{ABC}\ket{1}_D)}^2\geq 1-\frac{M^3\alpha'}{c\Delta},
\end{equation}
i.e., $\ket{\phi}_{ABC}$ has large overlap with $\ket{\psihist}$, as claimed.
At the same time, the following separate calculation shows that $\ket{\psihist}$ also has large overlap with $\ket{u}$:
\begin{equation}\label{eqn:step2}
    \abs{\bra{u}_{ABCD}(\ket{\psihist}_{ABC}\ket{1}_D)}^2=\frac{1}{2}\abs{\sum_{t=1}^{N}\frac{1}{\sqrt{N(M+1)}}}^2=\frac{N}{2(M+1)}\geq\frac{1}{4}+\Delta'.
\end{equation}
%By choosing $\Delta=M^5\alpha'$,
Via Equations (\ref{eqn:step1}) and (\ref{eqn:step2}) and the triangle inequality, we obtain our desired claim, i.e., that $\ket{\phi}_{ABC}\ket{1}_D$ has large overlap with $\ket{u}_{ABCD}$:
\begin{equation}\label{eqn:step3}
    \abs{(\bra{\phi}_{ABC}\bra{1}_D)\ket{u}_{ABCD}}^2\geq 1-\frac{1}{4}\left(2\sqrt{\frac{M^3\alpha'}{c\Delta}}+\sqrt{3-4\Delta'}\right)^2.
\end{equation}
Here, we have used the known identity $\trnorm{\ketbra{v}{v}-\ketbra{u}{u}}=2\sqrt{1-\abs{\braket{u}{v}}^2}$, and applied the triangle inequality to the trace norm.

Finally, we observe that all parameters in Equation (\ref{eqn:step3}) can be set so as to yield $\norm{\Pi_{H}\ket{u}}\ge \delta$.
Specifically, for sufficiently large fixed polynomial $\Delta$, the term $2\sqrt{\frac{M^3\alpha'}{c\Delta}}$ can be made arbitrarily inverse polynomially close to $0$. As for $\Delta'\in (0,1/4)$, choosing $\Delta'=1/4-1/q$ for some sufficiently large polynomial $q$, the term $\sqrt{3-4\Delta'}$ can be made to scale like $\sqrt{2+\epsilon}$ for $\epsilon>0$ any desired inverse polynomial. Taking into account the square, and combining with the $\Delta$ scaling, we conclude that for sufficiently large fixed polynomials $\Delta,\Delta'$, the right hand side above is at least $1-\frac{1}{4}(2+\epsilon')$ for any desired inverse polynomial $\epsilon'$. By choosing $\epsilon'$ small enough, this expression is at least $\delta^2$, from which it follows\footnote{The $\delta^2$ drops to $\delta$ here since we want a lower bound on $\norm{\Pi_{H}\ket{u}}$, not $\norm{\Pi_{H}\ket{u}}^2$.} that $\norm{\Pi_{H}\ket{u}}\ge \delta$, as desired. \\

 \noindent\emph{NO case.} Suppose $x\in \Pino$. We must show that $\lmin{H}\geq b$ (to satisfy the NO case requirement in $\GLHdes$) and $\norm{\Pi_{H}\ket{u}}\ge \delta$ (to satisfy the guiding state promise in $\GLHdes$).
 For the first claim, observe first that
 \begin{equation*}
    H'\succeq \hin+\hprop+\hstab+\hout\succeq \Omega\left(\frac{1-\sqrt{\beta}}{M^3}\right)=\beta',
 \end{equation*}
 where the first inequality follows since $\hin,\hprop,\hstab\succeq 0$ and $\Delta\geq 1$, and the second from \Cref{fact:1}. Since $\beta'>(\alpha'+\beta')/2$, we conclude that $\lmin{H}$ lies in the $0$-block relative
to $D$, implying
 \begin{equation*}
    \lmin{H}=\frac{\alpha'+\beta'}{2}=b,
 \end{equation*}
 as desired.
This also now yields the second claim regarding $\ket{u}$ --- by construction \emph{every} vector in the $0$-block is a ground state of $H$, including unit vector
\[
\ket{x}_A\ket{0\cdots 0}_{B}\left(\frac{1}{\sqrt{N}}\sum_{t=1}^{N}\ket{t}\right)_{C}\ket{0}_{D},
\]
 which has overlap $1/\sqrt{2}$ with $\ket{u}$, as claimed.
\end{proof}

%=============================================================
\section{Dequantizing the Quantum Singular Value Transformation}\label{sec:SVT}
%=============================================================
In this section we show how to dequantize quantum algorithms for the Singular Value Transformation when the input matrices are sparse, and when dealing with constant precision.
In Section \ref{sec:main} we present our main dequantization technique. In Section \ref{sec:singular} we apply this technique to construct a classical algorithm for estimating the singular values of a matrix given sampling-access to a state that has some overlap with the corresponding subspace. Finally, in Section \ref{sub:proofthclassical} we apply this singular value estimation algorithm to prove Theorem \ref{th:classical}.

In this section, the notation $\myO(\cdot)$ removes from the complexities  factors polynomial in $\log(MN)$, which mainly arise from the cost of accessing vectors and matrices using query-access and sampling-access as defined in Definitions~\ref{def:query-access}, \ref{def:sparsequery} and~\ref{def:sample}.

%=================================
\subsection{Main technique}\label{sec:main}
%=================================
In this subsection we describe our main technical tool. Concretely, we consider the following problem, which we denote by $\EST(s,\epsilon,\zeta)$, for an integer $s\ge 2$ and two parameters $\epsilon\in (0,1]$, $\zeta\in [0,1)$

\begin{center}
\fbox{
\begin{minipage}{15.1 cm} \vspace{2mm}

\noindent$\EST(s,\epsilon,\zeta)$\\\vspace{-3mm}

\noindent\hspace{3mm} Input: query-access to an $s$-sparse matrix $A\in\Comp^{M\times N}$ with $\norm{A}\le 1$%, where $m=O(\poly(n))$

\noindent\hspace{15mm}
query-access to a vector $u\in \Comp^{N}$ such that $\norm{u}\le 1$

\noindent\hspace{15mm}
$\zeta$-sampling-access to a vector $v\in \Comp^{N}$ such that  $\norm{v}\le 1$

\noindent\hspace{15mm}
even polynomial $P\in\Real[x]$ of degree $2d$ with $|P(x)|\le 1$ $~\forall x\in[-1,1]$
\vspace{1mm}

\noindent\hspace{3mm} Output: an estimate $\hat{z}\in \Comp$ such that $|\hat z-v^\dagger P(\sqrt{A^\dagger A})u|\le \epsilon$
\vspace{2mm}
\end{minipage}
}
\end{center}

We show the following result.
\begin{theorem}[Formal version of \Cref{th:classical-est}]\label{th:classical-est-formal}
For any $s\ge 2$ and any $\epsilon\in(0,1]$, the problem $\EST(s,\epsilon,\zeta)$ can be solved classically with probability at least $1-1/\poly(N)$ in $\myO(s^{2d}/\epsilon^2)$ time for any $\zeta\le \epsilon/8$.
%, where $\samp{v}$ denotes the time complexity of implementing one sampling-access query to $v$.
\end{theorem}

To prove Theorem \ref{th:classical-est-formal}, we will use the following lemma.
\begin{lemma}\label{lemma2}
Let $P\in\Real[x]$ be an even polynomial of degree $2d$. There is a $\myO(s^{2d})$-time classical procedure that given
\begin{itemize}
\item
query-access to an $s$-sparse matrix $A\in \Comp^{M\times N}$,
\item
query-access to a vector $u\in \Comp^{N}$,
\item
an index $i\in\{1,\ldots,N\}$,
\end{itemize}
outputs the $i$-th entry of $P(\sqrt{A^\dagger A})u$.
\end{lemma}
\begin{proof}
Let us write
\[
P(x)=a_0+a_2x^2+a_4x^4+\cdots+ a_{2d}x^{2d},
\]
where $2d$ is the degree of $P$ and $a_0,\ldots,a_{2d}\in\Real$. We have
\[
P(\sqrt{A^\dagger A})=a_0I_{2^n}+a_2A^\dagger A+a_4(A^\dagger A)^2+\cdots+a_{2d}(A^\dagger A)^d
\]
and thus
\[
P(\sqrt{A^\dagger A})u=a_0u+a_2A^\dagger Au+a_4(A^\dagger A)^2u+\cdots+a_{2d}(A^\dagger A)^du.
\]
We show below how to compute the $i$-th entry of each term $(A^\dagger A)^ru$, for $r\in\{1,\ldots,d\}$.

\paragraph*{The procedure.} Since $(A^\dagger A)^ru$ involves products of $r$ matrices, we give a more general recursive procedure that receives as input
\begin{itemize}
\item
query-access to $s$-sparse matrices $B^{[1]}$, \ldots, $B^{[r]}$, where $B^{[j]}\in\Comp^{M_j\times N_j}$ for each $j\in\{1,\ldots,r\}$, which satisfy the following conditions: $N_j=M_{j+1}$ for each $j\in\{1,\ldots,r-1\}$, and $N_{r}=N$;
\item
query-access to a vector  $u\in \Comp^{N}$;
\item
an index $i\in\{1,\ldots,N\}$,
\end{itemize}
and outputs the $i$-th entry of the vector $B^{[1]}\cdots B^{[r]} u$.

The procedure repeatedly applies the following idea to recursively decompose the product $B^{[1]}\cdots B^{[r]}$: to obtain the $i$th entry of $B^{[1]}\cdots B^{[r]}u$, it suffices to know only the $s$ non-zero entries of row $i$ of $B^{[1]}$ (which can be queried directly), together with the corresponding $s$ entries in column vector $B^{[2]}\cdots B^{[r]}u$ (which are computed recursively).

Formally, the procedure first queries all the entries of the $i$-th row of $B^{[1]}$, and gets at most $s$ non-zero entries and their positions, which we write $j_1,\ldots,j_\ell$, with $\ell\le s$.
The procedure then computes recursively the $j_q$-th entry of the product $B^{[2]}\cdots B^{[r]} u$, for each $q\in\{1,\ldots,\ell\}$.
Let $y_q$ denote the value obtained, for each $q\in\{1,\ldots,\ell\}$. The procedure finally outputs the value
\[
B^{[1]}_{i,j_1}y_1+\cdots+B^{[1]}_{i,j_\ell}y_\ell,
\]
 which is equal to the $i$-th entry of the vector $B^{[1]}\cdots B^{[r]} u$.

\paragraph*{Runtime.} Let $\Tt(r)$ denote the running time of this procedure. We have
\[
\Tt(r)\le s\cdot O(\poly(\log (MN)))+s\Tt(r-1)+s,
\]
 and thus $\Tt(r)=O(s^r\cdot \poly(\log(MN)))$.

The complexity of the whole procedure is
\[
O((s^2+\cdots+s^{2d})\cdot \poly(\log(MN)))=\myO(s^{2d}),
\]
as claimed.
\end{proof}

\begin{proof}[Proof of Theorem \ref{th:classical-est-formal}]
To compute an estimate of $v^\dagger P(\sqrt{A^\dagger A})u$ we will use the following approach (similar to Proposition 4.2 in \cite{TangSTOC19}). Remember that we have sampling-access to $v\in \Comp^{2^n}$. Let $(\Qq_v,\Sq_v,m)$ denote the triple from Definition \ref{def:sample}.
Let us write $w=P(\sqrt{A^\dagger A})u$. Consider the following procedure:
\begin{itemize}
\item[1.]
Using Algorithm $\Sq_v$, sample one index $j\in\{1,\ldots,N\}$;
\item[2.]
Using Algorithm $\Qq_v$, query the corresponding entry $v_j$;
\item[3.]
Compute the value $w_j$ using the algorithm of Lemma \ref{lemma2};
\item[4.]
Output the value $\frac{w_jm^2}{v_j}$.
\end{itemize}
Let $X$ denote the complex random variable corresponding to the output of this procedure. Let $\real(X)$ and $\comp(X)$ denote its real part and complex part, respectively.
We have
\begin{align*}
\E[X]&=\sum_{i=1}^{N} \frac{w_im^2}{v_i}p(i)
%=\sum_{i=1}^{N} v^\ast_iw_i=v^\dagger P(\sqrt{A^\dagger A}) u,
\end{align*}
where $m\in[(1-\zeta)\norm{v},(1+\zeta)\norm{v}]$ and
\[
p(i)\in\left[(1-\zeta)\frac{|v_i|^2}{\norm{v}^2},(1+\zeta)\frac{|v_i|^2}{\norm{v}^2}\right]
\]
for each $i\in\{1,\ldots,N\}$.
Note that the inequalities
\begin{align*}
m^2p(i)&\le (1+\zeta)^3 |v_i|^2\le (1+7\zeta)|v_i|^2,\\
m^2p(i)&\ge(1-\zeta)^3 |v_i|^2\ge (1-4\zeta)|v_i|^2,\\
\end{align*}
and thus also
\[
\abs{m^2p(i) - |v_i|^2}\le 7\zeta |v_i|^2
\]
hold for each $i\in\{1,\ldots,N\}$.
We thus have
\begin{align*}
\left|\E[X]-v^\dagger P(\sqrt{A^\dagger A}) u\right|=\left|\sum_{i=1}^{N} \frac{w_im^2}{v_i}p(i)-\sum_{i=1}^{N} v^\ast_iw_i
\right|
&=\left|\sum_{i=1}^{N} \frac{w_i}{v_i}\left(m^2p(i)-|v_i|^2\right)\right|\\
&\le\sum_{i=1}^{N} \left|\frac{w_i}{v_i}\right|
\cdot \left|m^2p(i)-|v_i|^2\right|\\
&\le
7\zeta\sum_{i=1}^{N} |w_i||v_i|\\
&\le
7\zeta \norm{w}\norm{v}\\
&\le7\zeta \norm{P(\sqrt{A^\dagger A})}\norm{u}\norm{v}\\
&\le7\zeta.
%\sum_{i=1}^{N} v^\ast_iw_i=v^\dagger P(\sqrt{A^\dagger A}) u,
\end{align*}
%and
%\begin{align*}
%\left|\E[\real(X)]-\real(v^\dagger P(\sqrt{A^\dagger A}) u)\right|&\le 7\zeta\\
%\left|\E[\comp(X)]-\real(v^\dagger P(\sqrt{A^\dagger A}) u)\right|&\le 7\zeta
%\end{align*}
%as well.
%We then have
%\begin{align*}
%\E[X]&=\sum_{i=1}^{N} \frac{w_i\norm{v}^2}{v_i}\frac{|v_i|^2}{\norm{v}^2}=\sum_{i=1}^{N} v^\ast_iw_i=v^\dagger P(\sqrt{A^\dagger A}) u,
%\end{align*}
%and thus $\E[\real(X)]=\real(v^\dagger P(\sqrt{A^\dagger A}) u)$ and $\E[\comp(X)]=\comp(v^\dagger P(\sqrt{A^\dagger A}) u)$. %We thus have either $\E[\real(X)]\ge b$ or $\E[\real(X)]\le a$.
We have
\begin{align*}
\var[\real(X)]&=\E[\real(X)^2]-\E[\real(X)]^2\\
&\le \E[\real(X)^2]\\
&\le\sum_{i=1}^{N} \left(\frac{|w_i|m^2}{|v_i|}\right)^2p(i)\\
&\le (1+7\zeta)\sum_{i=1}^{N} |w_i|^2m^2\\
&\le (1+7\zeta)(1+\zeta)^2\norm{w}^2\norm{v}^2\\
&\le (1+7\zeta)^2
\end{align*}
and similarly $\var[\comp(X)]\le (1+7\zeta)^2$. This procedure can be implemented in
\[
O(\poly(\log N))+\myO(s^{2d})=\myO(s^{2d})
\]
time.

Apply the above procedure $r$ times, each time getting a complex number $X_1$, $\ldots$, $X_r$, and output the mean. Let $Z=\frac{X_1+\cdots+X_r}{r}$ denote the corresponding complex random variable, and let  $\real(Z)$ and $\comp(Z)$ denote its real part and complex part, respectively. Since the variables $X_1$, $\ldots$, $X_r$ are independent, we have
\begin{align*}
\E[\real(Z)]&=\E[\real(X)],\\
%\E[\comp(Z)]&=\E[\comp(X)]=0,\\
\var[\real(Z)]&=\frac{1}{r^2}\left(\var[\real(X_1)]+\cdots+\var[\real(X_r)]\right) = \frac{\var[\real(X)]}{r}\le (1+7\zeta)^2/r,
\end{align*}
and similarly $E[\comp(Z)]=\E[\comp(X)]$ and $\var[\comp(Z)]\le (1+7\zeta)^2/r$.
By Chebyshev's inequality we obtain
\begin{align*}
\Pr\left[\big|\real(Z)-\E[\real(X)]\big|\ge \frac{\epsilon-7\zeta}{\sqrt{2}}\right]&\le \frac{2\cdot\var[\real(Z)]}{(\epsilon-7\zeta)^2}\le \frac{2(1+7\zeta)^2}{r(\epsilon-7\zeta)^2},\\
\Pr\left[\big|\comp(Z)-\E[\comp(X)]\big|\ge \frac{\epsilon-7\zeta}{\sqrt{2}}\right]&\le \frac{2\cdot\var[\comp(Z)]}{(\epsilon-7\zeta)^2}\le \frac{2(1+7\zeta)^2}{r(\epsilon-7\zeta)^2},
\end{align*}
and thus
\begin{align*}
\Pr\left[\big|Z-v^\dagger P(\sqrt{A^\dagger A}) u|\le \epsilon \right]&\ge\Pr\left[\big|Z-\E[X]|\le \epsilon-7\zeta \right]\\
&\ge 1-\Pr\left[\big|\real(Z)-\E[\real(X)]\big|\ge \frac{\epsilon-7\zeta}{\sqrt{2}}\right]\\
&\hspace{13mm}- \Pr\left[\big|\comp(Z)-\E[\comp(X)]\big|\ge \frac{\epsilon-7\zeta}{\sqrt{2}}\right]\\
&\ge 1-\frac{4(1+7\zeta)^2}{r(\epsilon-7\zeta)^2}.
\end{align*}
For $\zeta\le \epsilon/8$, taking $r=\Theta(1/\epsilon^2)$ guarantees that the estimate is an $\epsilon$-additive-approximation of $v^\dagger P(\sqrt{A^\dagger A}) u$ with probability at least $3/4$. Using Lemma \ref{lemma:powering} for each the real part and the imaginary part of the estimate, the success probability can then be improved to $1-1/\poly(N)$ by repeating this process $O(\log N)$ times.
\end{proof}

%======================================
\subsection{Singular Value Estimation by dequantizing the QSVT}\label{sec:singular}
%======================================
In this subsection we consider the following problem, which we denote  $\SV(s,t_1,t_2,\theta_1,\theta_2,\delta,\zeta)$, for an integer $s\ge 2$, real numbers $t_1,t_2,\theta_1,\theta_2\in[0,1]$ satisfying the inequalities $\theta_1\le t_1<t_2\le 1-\theta_2$, and parameters $\delta\in(0,1]$ and $\zeta\in[0,1)$.
This problem asks to decide if there exist singular values in a given interval, given sampling-access to a state that has some overlap with the corresponding subspace.

\begin{center}
\fbox{
\begin{minipage}{15 cm} \vspace{2mm}
\noindent$\SV(s,t_1,t_2,\theta_1,\theta_2,\delta,\zeta)$ \hspace{3mm}(singular value estimation) \\\vspace{-3mm}

\noindent\hspace{3mm} Input: query-access to an $s$-sparse matrix $A\in\Comp^{M\times N}$ with $\norm{A}\le 1$%, where $m=O(\poly(n))$

\noindent\hspace{15mm}
$\zeta$-sampling-access to a vector $u\in \Comp^{N}$ such that $\norm{u}\le 1$

\vspace{2mm}

\noindent\hspace{3mm}
Promise: (i) $A$ has a singular value in the interval $[t_1,t_2]$, and $\norm{\Pi^{[t_1,t_2]}_Au}\ge \delta$ holds, or

\noindent\hspace{19mm}
(ii) $A$ has no singular value in the interval $(t_1-\theta_1,t_2+\theta_2)$
\vspace{2mm}

\noindent\hspace{3mm} Goal: decide which of (i) or (ii) holds
\vspace{2mm}
\end{minipage}
}
\end{center}
Note that the condition about $u$ is only imposed in the case (i): for the case (ii) we do not require~$u$ to have any overlap with a specified subspace.

Our main result is the following theorem.
\begin{theorem}[Formal version of \Cref{th:SVT}]\label{th:SVT-formal}
For any $t_1,t_2,\theta_1,\theta_2\in[0,1]$ such that $\theta_1\le t_1<t_2\le 1-\theta_2$, any $\delta\in(0,1]$
and any $s\ge 2$, the problem $\SV(s,t_1,t_2,\theta_1,\theta_2,\delta,\zeta)$ can be solved classically with probability at least $1-1/\poly(N)$ in
\[
\myO\left(\frac{s^{c\cdot\left(1/\theta_1+1/\theta_2\right)\log(\sqrt{3}/\delta)}}{\delta^4}\right)
\]
 time, for some universal constant $c$, for any $\zeta\le \delta^2/56$.
\end{theorem}

In order to prove Theorem \ref{th:SVT-formal}, we will need the following lemma, which can be considered as a variant of Lemma 29 in \cite{Gilyen+STOC19}.

%\begin{lemma}\label{lemma1}
%Let $\tau,\chi\in(0,1/2)$ and $t_1,t_2\in[-1,1]$ such that $t_1<t_2$. There exists an even polynomial $P\in \Real[x]$ of degree $O(\log(1/\chi)/\tau)$ such that $|P(x)|\le 1$ for all $x\in[-1,1]$ and:
%\[
%\begin{cases}
%P(x)\in [1-\chi,1]& \textrm{ if } x\in[t_1,t_2],\\
%P(x)\in [0,\chi]& \textrm{ if } x\in[-1,t_1-\tau]\cup [t_2+\tau,1].\\
%\end{cases}
%\]
%\end{lemma}
\begin{lemma}\label{lemma1}
For any $t_1,t_2,\theta_1,\theta_2\in[0,1]$ such that $\theta_1\le t_1<t_2\le 1-\theta_2$ and any $\chi\in(0,1)$, there exists an even polynomial $P\in \Real[x]$ of degree
\[
O\left(\left(\frac{1}{\theta_1}+\frac{1}{\theta_2}\right)\log(1/\chi)\right)
\]
such that $|P(x)|\le 1$ for all $x\in[-1,1]$ and:
\[
\begin{cases}
P(x)\in [1-\chi,1]& \textrm{ if } x\in[t_1,t_2],\\
P(x)\in [0,\chi]& \textrm{ if } x\in[0,t_1-\theta_1]\cup [t_2+\theta_2,1].\\
\end{cases}
\]
\end{lemma}
\begin{proof}
Low and Chuang \cite{Low+17} (see also Lemma 25 in \cite{Gilyen+STOC19}) showed that for all $\eta>0$, $\xi\in(0,1/2)$, there exists an efficiently computable odd polynomial $P'\in\Real[x]$ of degree $O(\log(1/\xi)/\eta)$ such that:
\[
\begin{cases}
P'(x)\in [-1,1]& \textrm{ for all } x\in[-2,2],\\
P'(x)\in [-1,-1+\xi]& \textrm{ if } x\in[-2,-\eta],\\
P'(x)\in [1-\xi,1]& \textrm{ if } x\in[\eta,2].\\
\end{cases}
\]

Let $P_1'$ be the polynomial obtained by this construction with $\eta=\theta_1/2$ and $\xi=2\chi/5$, and $P_2'$ be the polynomial obtained by this construction with $\eta=\theta_2/2$ and $\xi=2\chi/5$. Define
\[
Q(x)=(1-\xi)\frac{P_1'(x-t_1+\theta_1/2)+P_2'(-x+t_2+\theta_2/2)}{2}+\xi.
\]
This polynomial satisfies the following properties:
\begin{itemize}
\item
for all $x\in [t_1,t_2]$,
\[
Q(x)\in\left[(1-\xi)\frac{1-\xi + 1-\xi}{2}+\xi, (1-\xi)\frac{1+1}{2}+\xi\right]\subseteq [1-\xi,1];
\]
\item
for all $x\in [0,t_1-\theta_1]\cup[t_2+\theta_2,1]$,
\[
Q(x)\in\left[(1-\xi)\frac{1-\xi - 1}{2}+\xi, (1-\xi)\frac{1+(-1+\xi)}{2}+\xi\right]\subseteq [0,3\xi/2].
\]
\end{itemize}
Additionally, using similar calculations, it can be checked that
$Q(x)\in[0,1]$ for all $x\in [0,1]$, and
$Q(x)\in[0,\xi]$ for all $x\in [-1,0]$.

Finally, set
\[
P(x)=\frac{Q(x)+Q(-x)}{1+\xi}.
\]
This is an even polynomial of degree $O\left(\left(1/\theta_1+1/\theta_2\right)\log(1/\chi)\right)$, which satisfies the conditions
\begin{itemize}
\item
$P(x)\in[0,1]$ for all $x\in[-1,1]$;
\item
for all $x\in [t_1,t_2]$,
\[
P(x)\in\left[\frac{1-\xi+0}{1+\xi}, \frac{1+\xi}{1+\xi}\right]\subseteq [1-2\xi,1]\subseteq[1-\chi,1];
\]
\item
for all $x\in [0,t_1-\theta_1]\cup[t_2+\theta_2,1]$,
\[
P(x)\in\left[\frac{0+0}{1+\xi}, \frac{3\xi/2+\xi}{1+\xi}\right]\subseteq [0,5\xi/2]= [0,\chi];
\]
\end{itemize}
as required.
\end{proof}

\begin{proof}[Proof of Theorem \ref{th:SVT-formal}]
Let $P$ denote the polynomial from Lemma \ref{lemma1} for the value  $\chi=\delta^2/3$. Observe that the degree of $P$ is $O((1/\theta_1+1/\theta_2)\log(\sqrt{3}/\delta))$.
%Observe that the degree of $P$ is constant since $\theta$ and $\delta$ are both constant.

Let
\[
A=\sum_{i=1}^{N} \sigma_i u_i v_i^\dagger
\]
be a singular value decomposition of $A$. Let $\Lambda_{[t_1,t_2]}$ be the set of indices $i\in\{1,\ldots,N\}$ such that $\sigma_i\in [t_1,t_2]$.

Let us decompose $u$ in the basis $\{v_i\}$:
\[
u=\sum_{i=1}^{N} \alpha_i v_i,
\]
where $\alpha_i\in\Comp$ for each $i\in\{1,\ldots,N\}$.
Remember that by assumption we have $\sum_{i=1}^{N}|\alpha_i|^2\le 1$.
We have
\begin{align*}
P(\sqrt{A^\dagger A})u&=\left(\sum_{i=1}^{N} P(\sigma_i) v_i v_i^\dagger \right)\left(\sum_{i=1}^{N} \alpha_i v_i\right)\\
&=\sum_{i=1}^{N} P(\sigma_i) \alpha_i v_i.
\end{align*}
Observe that $u^\dagger P(\sqrt{A^\dagger A}) u=\sum_{i=1}^{N} P(\sigma_i) |\alpha_i|^2$ is a real number.

If $A$ has a singular value in the interval $[t_1,t_2]$, then the promise guarantees that the inequality
\[
\sum_{i\in\Lambda_{[t_1,t_2]}}|\alpha_i|^2\ge \delta^2
\]
holds.
We thus have
\begin{align*}
u^\dagger P(\sqrt{A^\dagger A}) u&=\sum_{i=1}^{N} P(\sigma_i) |\alpha_i|^2
\ge \sum_{i\in \Lambda_{[t_1,t_2]}} P(\sigma_i) |\alpha_i|^2
\ge \sum_{i\in \Lambda_{[t_1,t_2]}} (1-\chi)|\alpha_i|^2
\ge (1-\chi)\delta^2\ge 2\delta^2/3.
\end{align*}

If $A$ has no singular value in the interval $(t_1-\theta_1,t_2+\theta_2)$, then we have
\begin{align*}
u^\dagger P(\sqrt{A^\dagger A})  u&=\sum_{i=1}^{N} P(\sigma_i) |\alpha_i|^2
\le \sum_{i=1}^{N} \chi|\alpha_i|^2\le \chi \le \delta^2/3.
\end{align*}

Using the algorithm $\EST(s,\epsilon,\zeta)$ on input $(A,u,u,P)$ with $\epsilon=\delta^2/7$,
Theorem~\ref{th:classical-est-formal} guarantees that for $\zeta\le \epsilon/8=\delta^2/56$,
we can distinguish with probability at least $1-1/\poly(N)$ between the two cases in
\[
\myO\left(\frac{s^{O((1/\theta_1+1/\theta_2)\log(\sqrt{3}/\delta))}}{\delta^4}\right)
\]
time.
\end{proof}

%=====================================
\subsection{Classical estimation of the ground state energy of local Hamiltonians}\label{sub:proofthclassical}
%=====================================
In this subsection we show how the results from Section \ref{sec:singular} can be used to prove Theorem \ref{th:classical}.

First, we introduce a ``gapped decision version'' of the problem $\GLHest$, which we write $\GLH(k,a,b, \delta)$, for an integer $k\ge2$, two real numbers $a,b\in[-1,1]$ such that $a<b$, and $\delta\in(0,1]$. This problem is defined as follows.
%We actually consider two versions of this problem: an estimation version that we denote $\GLHest(k,\epsilon, \delta,\epsilon)$ and a promise decision version that we denote $\GLH(k,a,b, \delta,\epsilon)$, where $k\ge 2$ is an integer and $a,b,\epsilon,\delta$ are real numbers in the interval $(0,1)$. The description of these two problems is as follows.

\begin{center}
\fbox{
\begin{minipage}{13 cm} \vspace{2mm}

\noindent$\GLH(k,a,b,\delta)$ \hspace{3mm}(guided local Hamiltonian, gapped decision version) \\\vspace{-3mm}

\noindent\hspace{3mm} Input: a $k$-local Hamiltonian $H$ acting on $n$ qubits such that $\norm{H}\le 1$

\noindent\hspace{15mm}
%the classical description of a vector $u\in \Sample_n$ such that $\norm{u}\le 1$ and $\norm{\Pi_Hu}\ge \delta$
sampling-access to a vector $u\in\Comp^{2^n}$ such that $\norm{u}\le 1$
\vspace{2mm}

\noindent\hspace{3mm}
Promises: (i) $\norm{\Pi_Hu}\ge \delta$

\noindent\hspace{21mm} (ii) either $\lambda_{H}\le a$ or $\lambda_{H}\ge b$ holds
\vspace{2mm}

\noindent\hspace{3mm} Goal: decide which of $\lambda_{H}\le a$ or $\lambda_{H}\ge b$ holds
\vspace{2mm}
\end{minipage}
}
\end{center}
%\snote{for both defs, should first mention that by ``classical description'' of $u$, we mean \Cref{def:query-access} or \Cref{def:sample}.}
%\fnote{I agree (I will polish the writing and the explanations later once we know exactly which definition we need to prove the hardness result). ``Classical description'' should refer to \Cref{def:sample}. Note that for the proof of $\BQP$-hardness this does not matter much, since the state that has large overlap with the ground state should be very easy to describe classically.}
%Note that $\GLH(k,a,b,\delta)$ can be obviously solved by solving $\GLHest(k,(b-a)/2,\delta)$.

The following result follows easily from Theorem \ref{th:SVT-formal}.
\begin{proposition}\label{prop:classical-decision}
For any constants $a,b\in[-1,1]$ such that $a<b$, any constant $\delta\in(0,1]$ and any $k=O(\log n)$, the problem $\GLH(k,a,b,\delta)$ can be solved classically with probability at least $1-1/\exp(n)$ in $\poly(n)$ time.
\end{proposition}
\begin{proof}
Let $(H,u)$ denote the input of the problem $\GLH(k,a,b,\delta)$. Since $H$ is a $k$-local Hamiltonian, i.e., $H$ is a sum of $m=\poly(n)$ terms, where each term is a $k$-local Hermitian matrix, $H$ is $m2^k$-sparse. Since $\norm{H}\le 1$, all the eigenvalues of $H$ are in the interval $[-1,1]$. The Hermitian matrix $\frac{H+3I}{4}$ is thus $(m2^k+1)$-sparse, and has all its eigenvalues in the interval $[1/2,1]$.
In order to solve $\GLH(k,a,b,\delta)$, we can simply use the algorithm of Theorem \ref{th:SVT-formal} for $\SV(m2^k+1,1/2,(3+a)/4,1/2,(b-a)/4,\delta,0)$ on input $(\frac{H+3I}{4},u)$.
\end{proof}

We then show how Theorem \ref{th:classical} follows from Proposition \ref{prop:classical-decision}.
\begin{proof}[Proof of Theorem \ref{th:classical}]
Let $(H,u)$ denote the input of the problem $\GLHest(k,\epsilon,\delta)$.

For some integer $r=O(1)$ that will be fixed  later, decompose the interval $[-1,1]$ into $2r$ intervals, each of length $1/r$. The $i$-th interval is $[\frac{i-r-1}{r},\frac{i-r}{r}]$, for each $i\in\{1,\ldots,2r\}$.

For each $i\in\{1,\ldots,2r\}$, we apply the algorithm of Proposition \ref{prop:classical-decision} (which we denote below by Algorithm $\Cc$) to solve the problem $\GLH(k,a_i,b_i,\delta)$ on input $(H,u)$ with $a_i=\frac{i-r-1}{r}$ and $b_i=\frac{i-r}{r}$. With probability at least $1-1/\exp(n)$, Algorithm~$\Cc$ never errs during the $2r$ iterations. We assume below that this is the case. Under this assumption, for each $i$ we have:
\begin{itemize}
\item
 if $\lambda_H\le a_i$ then $\Cc$ outputs ``$\lambda_H\le a_i$'',
\item
if $\lambda_H\ge b_i$ then $\Cc$ outputs ``$\lambda_H\ge b_i$'',
\item
if $\lambda_H\in(a_i,b_i)$ then the output of $\Cc$ is arbitrary (i.e., the output is either ``$\lambda_H\le a_i$'' or ``$\lambda_H\ge b_i$'').
\end{itemize}

One of the following three cases necessarily occurs:
\begin{itemize}
\item[(a)] Algorithm $\Cc$ outputs ``$\lambda_H\le a_i$'' for all iterations $i\in\{1,\ldots,2r\}$,
\item[(b)] Algorithm $\Cc$ outputs ``$\lambda_H\ge b_i$'' for all iterations $i\in\{1,\ldots,2r\}$,
\item[(c)] there is one value $i_0\in\{1,\ldots,2r\}$ such that Algorithm $\Cc$ outputs ``$\lambda_H\ge b_i$'' for all iterations $i=1,\ldots,i_0$ and outputs ``$\lambda_H\le a_i$'' for all iterations $i=(i_0+1),\ldots, 2r$.
\end{itemize}
In Case (a) we can conclude that $\lambda_H\in [-1,-1+1/r]$. In Case (b) we can conclude that $\lambda_H\in [1-1/r,1]$. In Case (c) we can conclude that $\lambda_H\in [\frac{i_0-r-1}{r},\frac{i_0-r+1}{r}]\cap[-1,1]$. Taking $r=2/\epsilon$ thus guarantees that we get an $\epsilon$-additive-approximation of $\lambda_H$.
\end{proof}

\begin{remark}
While the problem $\GLHest$ stated in the introduction and the problem $\GLH$  stated above assume (for simplicity) perfect sampling-access to the vector $u$, i.e., $0$-sampling-access, Theorem \ref{th:classical} also holds for the version of the problem where $u$ is given via $\zeta$-sampling-access with $\zeta\le \delta^2/56$, since the algorithm of Theorem \ref{th:SVT-formal} works for such sampling-access as well.
\end{remark}
%=====================================================
\section{Implications to the Quantum PCP Conjecture}\label{sec:PCP}
%=====================================================
In this section we discuss the implications of our dequantization results to the quantum PCP conjecture.

We first define a notion of quantum states that have a short classical description from which efficient sampling-access to the state can be simulated. More precisely, we consider families of complex vectors that can be represented by short binary strings that allow to efficiently generate the triple $(\Qq_u, \Sq_u, m)$ of Definition~\ref{def:sample}.
\begin{definition}\label{def:description}
%Consider a family of vectors $\{u^\ell\}_{\ell\in\Nat}$, where $u_\ell\in\Comp^{N_\ell}$ for each $\ell\in\Nat$.
For a parameter $\zeta\in[0,1]$,
we say that a family of complex vectors $\Ff=\{u^{(\ell)}\}_{\ell\in\Nat}$ has a succinct representation allowing $\zeta$-sampling-access if there exist an injective function $\rep\colon\Ff\to\{0,1\}^\ast$, a polynomial $q$, and a polynomial-time algorithm $\Aa$ satisfying the following conditions.
\begin{itemize}
\item
For all $u\in\Ff$, the length of the binary string $\rep(u)$ is at most $q(\log (\dim u))$, where $\dim u$ denotes the dimension of $u$.
\item
Algorithm $\Aa$ receives as input a binary string. If the string does not correspond to the string $\rep(u)$ for some $u\in\Ff$, then Algorithm $\Aa$ outputs an error message. Otherwise Algorithm $\Aa$ on input $\rep(u)$ outputs the classical description\footnote{By a classical representation of $\Qq_{u}$ and $\Sq_{u}$ we mean the description of $\poly(\log \dim u)$-size circuits that implement these two algorithms.} of Algorithms $\Qq_{u}$ and $\Sq_{u}$ and a real number $m$ satisfying the specifications in Definition~\ref{def:sample}.
\end{itemize}
If $\zeta=0$, we say that the family has a succinct representation allowing perfect-sampling-access.
\end{definition}

Let $\zeta\ge 0$ be a constant. Consider a family of complex vectors $\Ff$ that has a succinct representation allowing $\zeta$-sampling-access. For any $n\ge 1$, we define
\[
\Ff_n=\big\{ u\in\Ff\:|\:\dim u=2^n\big\},
\]
the set of vectors of $\Ff$ that are $2^n$-dimensional. We now define a new version of the $k$-local Hamiltonian problem in which we are guaranteed that there exists a vector in $\Ff$ that has some overlap with the ground state of the Hamiltonian. We denote this problem by $\LH_\Ff(k,a,b,\delta)$, for an integer $k\ge2$, two real numbers $a,b\in[-1,1]$ such that $a<b$, and $\delta\in(0,1]$. Note that while~$\zeta$ does not appear explicitly in the notation, the problem $\LH_\Ff(k,a,b,\delta)$ does depend on this parameter since $\Ff$ depends on $\zeta$.
\begin{center}
\fbox{
\begin{minipage}{15.5 cm} \vspace{2mm}

\noindent$\LH_\Ff(k,a,b,\delta)$ \hspace{3mm}($k$-local Hamiltonian problem with samplable approximate ground state)\\\vspace{-3mm}

\noindent\hspace{3mm} Input: a $k$-local Hamiltonian $H$ acting on $n$ qubits such that $\norm{H}\le 1$

\noindent\hspace{3mm}
Promise: (i) there exists a vector $u\in\Ff_n$ such that $\norm{\Pi_H u}\ge \delta$

\noindent\hspace{19mm}
(ii) $\lambda_{H}\le a$, or $\lambda_{H}\ge b$ holds
\vspace{2mm}

\noindent\hspace{3mm} Goal: decide which of $\lambda_{H}\le a$ or $\lambda_{H} \ge b$ holds
\vspace{2mm}
\end{minipage}
}
\end{center}

A consequence of our dequantization result is the following result.

\begin{theorem}[Formal version of \Cref{th:qPCP}]\label{th:qPCP-formal}
For any $k=O(\log n)$, any $a,b\in[-1,1]$ such that $b-a=\Omega(1)$, and any constant $\delta>0$ such that $\delta\ge \sqrt{56\zeta}$, the problem $\LH_\Ff(k,a,b,\delta)$ is in $\MA$.
\end{theorem}
Note that when $\zeta=0$, i.e., when the vectors in the family $\Ff$ have a succinct representation allowing perfect sampling-access, the only condition on $\delta$ is that $\delta$ should be a positive constant.
\begin{proof}[Proof of Theorem \ref{th:qPCP-formal}]
We show below that the problem can be solved with high probability by a $\poly(n)$-time classical verifier that receives a classical witness of length $\poly(n)$.  Note that since $H$ is a $k$-local Hamiltonian, i.e., $H$ is a sum of $m=\poly(n)$ terms  where each term is a $k$-local Hermitian matrix, $H$ is $m2^k$-sparse.

Consider the case where $\lambda_H\le a$. The witness in this case is simply the classical representation $\rep(u)$ of a vector $u\in\Ff_n$ such that $\norm{\Pi_H u}\ge \delta$. The verifier checks that the representation is valid using Algorithm~$\Aa$ from Definition \ref{def:description}. If the representation is not valid, then it rejects. Otherwise, the verifier applies the polynomial-time algorithm of Theorem \ref{th:SVT-formal} (which we denote below by Algorithm~$\Bb$) to solve the problem $\SV(m2^k+1,1/2,(3+a)/4,1/2,(b-a)/4,\delta,\zeta)$ on input $(\frac{H+3I}{4},u)$ using the triple $(\Qq_u,\Sq_u,m)$ obtained by Algorithm~$\Aa$. The verifier  accepts if Algorithm $\Bb$ outputs that $\frac{H+3I}{4}$ has a singular value in the interval $[1/2,(3+a)/4]$, and rejects otherwise. Theorem~\ref{th:SVT-formal} says the verifier will accept with probability $1-1/\exp(n)$.

Consider the case where $\lambda_H\ge b$. For any binary string received as witness, the verifier will reject unless the witness is of the form $\rep(u)$ for some vector $u\in\Ff_n$. If the witness is of this form, Theorem \ref{th:SVT-formal} then guarantees that Algorithm $\Bb$ outputs with probability at least $1-1/\exp(n)$ that $\frac{H+3I}{4}$ has no singular value in the interval $(0,(3+b)/4)$, in which case the verifier rejects.
\end{proof}

Finally, using Definition \ref{def:description}, we give the formal version of the conjecture introduced in Section \ref{sub:app}.

\begin{conjecture}[NLSS conjecture --- formal version of \Cref{conj:variant}]\label{conj:variant-formal}
There exist a family of $O(1)$-local Hamiltonians $\{H_n\}_{n\in\Nat}$, where each $H_n$ acts on $n$ qubits, and a constant $\epsilon>0$ such that for any family of complex vectors $\Ff$ that has a succinct representation allowing perfect-sampling-access, we have for any sufficiently large~$n$:
\[
u^\dagger H_n u >\lambda_{H_n}+\epsilon
\]
for any unit-norm vector $u\in\Ff_n$.
\end{conjecture}

\section*{Acknowledgments}
The authors are grateful to Ryan Babbush, Dominic Berry and Andr{\'a}s Gily{\'e}n for helpful comments about the manuscript, and to anonymous referees for their feedback. FLG was supported by JSPS KAKENHI grants JP19H04066, JP20H05966, JP20H00579, JP20H04139, JP21H04879 and MEXT Quantum Leap Flagship Program (MEXT Q-LEAP) grants JPMXS0118067394, JPMXS0120319794. SG was supported by DFG grants 450041824 and 432788384.

%=====================================
%=====================================
\bibliographystyle{alpha}
\bibliography{Qchemistry}
%=====================================
\end{document}